%% file: neurips_2022.tex
\documentclass{article}

\usepackage{hyperref}
\usepackage{url}
\usepackage{packages}
\usepackage{commands}
\usepackage{mymacros}

\usepackage[final, nonatbib]{neurips_2022}

\usepackage[utf8]{inputenc} 
\usepackage[T1]{fontenc}    
\usepackage{hyperref}       
\usepackage{url}            
\usepackage{booktabs}       
\usepackage{amsfonts}       
\usepackage{nicefrac}       
\usepackage{microtype}      
\usepackage{xcolor}         
\usepackage[
backend=biber,
style=numeric,
citestyle=numeric,
sorting=none
]{biblatex}

\title{Zero-Sum Stochastic Stackelberg Games}

\author{Denizalp Goktas 
\\
Department of Computer Science\\
Brown University\\
Providence, RI 02906, USA \\
\texttt{denizalp\_goktas@brown.edu} \\
\And
Jiayi Zhao \\
Department of Computer Science \\
Pomona College \\
Pomona, CA, USA \\
\texttt{jzae2019@mymail.pomona.edu} \\
\AND
Amy Greenwald \\
Brown University\\
Providence, RI 02906, USA \\
\texttt{amy\_greenwald@brown.edu}
}

\begin{document}

\maketitle

\input{abstract}

\input{intro}
\input{related}
\input{prelim}
\input{properties}

\input{subdifferential}
\input{fisher}
\input{experiments}
\input{conclusion}
\input{ack}
\printbibliography
\newpage
\appendix
\input{checklist}
\newpage
\input{appendix/algo}

\input{appendix/properties}
\input{appendix/benveniste}
\input{appendix/fisher}
\input{appendix/experiments}

\end{document}

%% file: abstract.tex
\begin{abstract}
    Zero-sum stochastic games have found important applications in a variety of fields, from machine learning to economics. Work on this model has primarily focused on the computation of Nash equilibrium due to its effectiveness in solving adversarial board and video games. Unfortunately, a Nash equilibrium is not guaranteed to exist in zero-sum stochastic games when the payoffs at each state are not convex-concave in the players' actions. A Stackelberg equilibrium, however, is guaranteed to exist. Consequently, in this paper, we study zero-sum stochastic Stackelberg games. Going beyond known existence results for (non-stationary) Stackelberg equilibria, we prove the existence of recursive (i.e., Markov perfect) Stackelberg equilibria (recSE) in these games, provide necessary and sufficient conditions for a policy profile to be a recSE, and show that recSE can be computed in (weakly) polynomial time via value iteration. Finally, we show that zero-sum stochastic Stackelberg games can model the problem of pricing and allocating goods across agents and time. More specifically, we propose a zero-sum stochastic Stackelberg game whose recSE correspond to the recursive competitive equilibria of a large class of stochastic Fisher markets. We close with a series of experiments that showcase how our methodology can be used to solve the consumption-savings problem in stochastic Fisher markets.
\end{abstract}

%% file: intro.tex
Min-max optimization has paved the way for recent progress in a variety of fields, from machine learning to economics.
In a \mydef{constrained min-max optimization problem},
$\min_{\outer \in \outerset} \max_{\inner \in \innerset} \obj(\outer, \inner)$, the objective function $\obj: \outerset \times \innerset \to \R$ is continuous, and the constraint sets $\outerset \subset \R^\outerdim$ and $\innerset \subset \R^\innerdim$ are nonempty and compact.
When $\obj$ is convex-concave, and the constraint sets $\outerset$ and $\innerset$ are convex, the seminal minimax theorem \cite{neumann1928theorie, sion1958general} holds, i.e., $\min_{\outer \in \outerset} \max_{\inner \in \innerset} \obj(\outer, \inner) = \max_{\inner \in \innerset} \min_{\outer \in \outerset} \obj(\outer, \inner)$, and such problems can be interpreted as solving a \mydef{min-max} (or \mydef{zero-sum}) one-shot \emph{simultaneous-move\/} \mydef{game} 
between an \mydef{outer player} $\outer$ and an \mydef{inner player} $\inner$ with respective payoff functions $-\obj$, $\obj$ and respective action sets $\outerset$, $\innerset$, where the solutions $(\outer^*, \inner^*) \in \outerset \times \innerset$ 
are \mydef{Nash equilibria}: i.e., \mydef{best responses} to one another with $\outer^* \in \argmin_{\outer \in \outerset} \obj(\outer, \inner^*)$ and $\inner^* \in \argmax_{\inner \in \innerset} \obj(\outer^*, \inner)$.

More generally, one can consider \mydef{zero-sum stochastic games} played over an infinite discrete time horizon $\N_+$.
The game starts at some initial state $\staterv[0] \sim \initstates$. 
At each subsequent time-step $\iter \in \N_+$, players encounter a new state $\state[\iter] \in \states$.
After taking their respective actions $(\outer[\iter], \inner[\iter])$ from their respective action spaces $\outerset(\state[\iter]) \subseteq \R^\outerdim$ and $\innerset(\state[\iter]) \subseteq \R^\innerdim$,
they receive payoffs $\reward (\state[\iter], \outer[\iter], \inner[\iter])$, and then either transition to a new state $\staterv[\iter+1] \sim \trans (\cdot \mid \state[\iter], \outer[\iter], \inner[\iter])$ with probability $\discount$, or the game ends with the remaining probability.
The goal of the outer (resp.{} inner) player is to play  
a (\mydef{non-stationary}) \mydef{policy}, i.e., a sequence of history-dependent actions $\{\outer[\iter]\}_{\iter = 0}^\infty$ (resp.{} $\{\inner[\iter]\}_{\iter = 0}^\infty$),
that maximizes (resp.{} minimizes) their expected cumulative discounted payoff (resp. loss) $\mathop{\Ex} \left[\sum_{\iter = 0}^\infty \discount^\iter \reward(\staterv[\iter], \outer[\iter], \inner[\iter]) \right]$, fixing their opponent's policy. 

Note that any zero-sum stochastic game can be (non-compactly) represented as a zero-sum one-shot game with objective function $\mathop{\Ex} \left[\sum_{\iter = 0}^\infty \discount^\iter \reward(\staterv[\iter], \outer[\iter], \inner[\iter]) \right]$ and action sets comprising sequences $ \prod_{t=0}^\infty \outerset$ and $ \prod_{t=0}^\infty \innerset$ for the outer and inner players, respectively.
Although the traditional minimax theorem \cite{neumann1928theorie, sion1958general} does not hold in these games, because this objective function is not convex-concave in these actions, \citeauthor{shapley1953stochastic} nonetheless establishes the existence of Nash equilibria \cite{shapley1953stochastic},
by deriving a more general minimax theorem.

A \mydef{stationary} (or \mydef{Markovian~\cite{maskin2001markov}}) \mydef{policy} is a mapping from states to actions.
When $\reward (\state, \outer, \inner)$ is bounded, continuous, and concave-convex in $(\outer, \inner)$, for all $\state \in \states$, 
we are guaranteed the existence of a stationary policy profile, i.e., a pair of stationary policies $\policy[\outer]: \states \to \outerset$, $\policy[\inner]: \states \to \innerset$ for the outer and inner players, respectively, specifying the actions taken at each state, with a unique value such that both players maximize their expected payoffs, as a generalization of the minimax theorem holds \cite{shapley1953stochastic}:%
\footnote{Shapley's original results, which concern state-dependent payoff functions that are bilinear in the outer and inner players' actions, extend directly to payoffs which are convex-concave in the players' actions.}
\vspace{-5mm}

{\small
\begin{align*}
\min_{\policy[\outer] \in \outerset^\states} \max_{\policy[\inner] \in \innerset^\states} \mathop{\E} \left[\sum_{\iter = 0}^\infty \discount^\iter \reward(\staterv[\iter], \policy[\outer](\staterv[\iter]), \policy[\inner](\staterv[\iter])) \right]
&= \max_{\policy[\inner] \in \innerset^\states} \min_{\policy[\outer] \in \outerset^\states}  \mathop{\E} \left[\sum_{\iter = 0}^\infty \discount^\iter \reward(\staterv[\iter], \policy[\outer](\staterv[\iter]), \policy[\inner](\staterv[\iter])) \right]
\end{align*}}

\vspace{-2mm}
In other words, under the aforementioned assumptions, we are guaranteed the existence of a \mydef{recursive Nash equilibrium} (sometimes called Markov perfect equilibrium \cite{maskin2001markov}, a refinement of extensive-form game subgame perfect equilibrium for Markov games), namely a stationary policy profile in which players not only play Nash equilibria: i.e., best responses to one another, but do so regardless of initial state.
Additionally, when the rewards at each state are convex-concave, a recursive Nash equilibrium can be computed in polynomial time by iterative application of the min-max operator  \cite{shapley1953stochastic}.
Zero-sum \emph{stochastic\/} games generalize zero-sum games from a single state to multiple states, and have found even more applications in a variety of fields \cite{jaskiewicz2018non}. 

When the objective function in a min-max optimization problem is not convex-concave, a minimax theorem is not guaranteed to hold, precluding the interpretation of the game as simultaneous-move, and the guaranteed existence of Nash equilibrium.
Nonetheless, the optimization problem can still be viewed as a min-max (or zero-sum) one-shot \emph{Stackelberg\/} game, in which the outer player moves before the inner one.
The canonical solution concept in such games is \mydef{Stackelberg equilibrium (SE)}.
Interestingly, in Stackelberg games, the inner player's actions can be constrained by the outer player's choice, without impacting existence~\cite{goktas2021minmax}.
The result is a \mydef{min-max one-shot Stackelberg game}: i.e., $\min_{\outer \in \outerset} \max_{\inner \in \innerset: h(\outer, \inner) \geq \zeros} \obj (\outer, \inner)$, where $\obj, h: \outerset \times \innerset \to \R$ are continuous, and $\outerset, \innerset$ are non-empty and compact.
Even more problems of interest are captured by this model \cite{bertsimas2011theory,  fisac2015reach, bansal2017hamilton, leung2022learning}.

One can likewise consider \mydef{zero-sum stochastic Stackelberg games}, which generalize both zero-sum one-shot Stackelberg games and zero-sum stochastic (simultaneous-move) games.
Similar to zero-sum stochastic games, 
these games are played over an infinite discrete time horizon $\N_+$, start at some state $\staterv[0] \sim \initstates$ and consist of nonempty and compact actions spaces $\outerset \subset \R^\outerdim$ and $\innerset \subset \R^\innerdim$%
\footnote{To simplify notation, we drop the dependency of action spaces on states going forward, but our theory applies in this more general setting.},
a state-dependent payoff function
$\reward (\state, \outer, \inner)$, a transition probability $\trans(\state^\prime \mid \state, \outer, \inner)$, and a discount rate $\discount$.
In addition, they are augmented with a state-dependent (joint action) constraint function $\constr(\state, \outer, \inner)$,
with two players that seek to optimize their cumulative discounted payoffs, in expectation, while satisfying the constraint $\constr(\state, \outer, \inner) \geq \zeros$ at each state $\state \in \states$.
Applications of this model include autonomous driving \cite{fisac2015reach, leung2022learning}, reach-avoid problems in human-robot interaction \cite{bansal2017hamilton}, and robust optimization in stochastic environments \cite{bertsimas2011theory}, 
\deni{Double check if there are more papers to add!.}
and, as we show, economic markets.

\deni{Here is the new story: at a high level, while in a stochastic game, players pick policies *before the start of the game* simultaneously, in stochastic Stackelberg games, a leader picks a policy first, and then the follower pick a best response policy *before the start of the game*. Our goal then becomes to investigate if playing stochastic Stackelberg game, can be seen as equivalent to playing a Stackelberg game at each state of the game, i.e., does \recSE{} or Markov perfect Stackelberg equilibrium exist.}

While in stochastic games, players announce their policies simultaneously before play commences, in stochastic Stackelberg games, the outer player, 
announces their policy first, after which the inner player announces theirs.
The canonical solution concept for such games is the \mydef{Stackelberg equilibrium}, which is guaranteed to exist in non-stationary policies under mild assumptions \cite{goktas2021minmax,vorobeychik2012computing}.
The computational complexity of such policies, however, can be prohibitive, since even representing such policies in an infinite horizon setting is intractable.
A natural question to ask, then, is whether \emph{stationary\/} equilibria exist in zero-sum Stackelberg games, i.e., stationary policy profiles in which, starting at any state of the game, the outer player maximizes their expected discounted cumulative payoff while the inner player best responds.
Following the analogous Nash nomenclature, we call such policies \mydef{recursive Stackelberg equilibria (\recSE{})} (or Markov perfect Stackelberg equilibria).

In this paper, we define and prove the existence of \recSE{} in zero-sum stochastic Stackelberg games, provide necessary and sufficient conditions for a policy profile to be a \recSE, and show that a \recSE{} can be computed in (weakly) polynomial time via value iteration.
We further show that zero-sum stochastic Stackelberg games can be used to solve problems of pricing and allocating goods across agents and time.
In particular, we introduce \mydef{stochastic Fisher markets}, a stochastic generalization of the Fisher market \cite{brainard2000compute}, and a special case of \citeauthor{Friesen1979stochastic}'s \cite{Friesen1979stochastic} financial market model, which itself is a stochastic generalization of the \citeauthor{arrow-debreu} model of a competitive economy \cite{arrow-debreu}.
We then prove the existence of recursive competitive equilibrium (recCE) \cite{mehra1977recursive} in this model, under the assumption that consumers have continuous and homogeneous utility functions, by characterizing the recCE of any stochastic Fisher market as the recSE of a corresponding zero-sum stochastic Stackelberg game.
Finally, we use value iteration to solve various stochastic Fisher markets, highlighting the issues that value iteration can encounter, depending on the smoothness properties of the
utilities.

%% file: related.tex
\paragraph{Related Work}

Algorithms for min-max optimization problems (i.e., zero-sum games) with independent strategy sets have been extensively studied 
\cite{tseng1995variational, nesterov2006variational, gidel2020variational, mokhtari2020convergence, ibrahim2019lower, zhang2020lower, lin2020near, alkousa2020accelerated, juditsky2011first, hamedani2018primal, zhao2019optimal, thekumparampil2019efficient, ouyang2018lower, nemirovski2004prox, nesterov2007dual, tseng2008accelerated, sanjabi2018stoch, nouiehed2019solving, lu2019block, jin2020local, ostrovskii2020efficient, lin2020gradient, zhao2020primal, rafique2019nonconvex} (for a summary see, Section G \cite{goktas2021minmax}).
\citeauthor{goktas2021minmax} studied min-max games with dependent strategy sets, proposing polynomial-time nested gradient descent ascent (GDA) \cite{goktas2021minmax} and simultaneous GDA algorithms for such problems \cite{goktas2021robustminmax}.

The computation of Stackelberg equilibrium in two-player stochastic Stackelberg games has been studied in several interesting settings, in which the leader moves before the follower, but without the leader's actions impacting the followers' choice sets.
\citeauthor{bensoussan2015maximum} \cite{bensoussan2015maximum} study continuous-time general-sum stochastic Stackelberg games with continuous action spaces, and prove existence of a solution in this setting.
\citeauthor{vorobeychik2012computing} \cite{vorobeychik2012computing} consider a general-sum
stochastic Stackelberg game with finite state-action spaces and an infinite horizon.
These authors show that stationary SE policies do not exist in this very general setting, but nonetheless identify a subclass of games, namely team (or potential) Stackelberg games for which stationary Stackelberg equilibrium policies do exist.
\citeauthor{vu2022stackelberg} \cite{vu2022stackelberg} study the empirical convergence of policy gradient methods in the same setting as \citeauthor{vorobeychik2012computing} \cite{vorobeychik2012computing}, while \citeauthor{ramponi2022learning} \cite{ramponi2022learning} study non-stationary equilibria in this same setting, assuming a finite horizon.
\citeauthor{chang2015leader} \cite{chang2015leader} and \citeauthor{sengupta2020multi} \cite{sengupta2020multi} consider a partially observable version of \citeauthor{vorobeychik2012computing}'s  \cite{vorobeychik2012computing} model, and provide methods to compute Stackelberg equilibria in their setting.


Some recent research concerns one leader-many followers Stackelberg games.
\citeauthor{vasal2020stochastic} \cite{vasal2020stochastic} studies a discrete-time, finite horizon one leader-many follower stochastic Stackelberg game
with discrete action and state spaces, and provides algorithms to solve such games.
\citeauthor{demiguel2009stochastic} \cite{demiguel2009stochastic} consider a stochastic Stackelberg game-like market model with $n$ leaders and $m$ followers; they prove the existence of a SE in their model, and provide (without theoretical guarantees) algorithms that converge to such an equilibrium in experiments.
Dynamic Stackelberg games \cite{li2017review} have been applied to a wide range of problems, including security \cite{vasal2020stochastic, vorobeychik2012computing}, insurance provision \cite{chen2018new, yuan2021robust}, advertising \cite{he2008cooperative}, robust agent design \cite{rismiller2020stochastic}, allocating goods across time intertemporal pricing \cite{oksendal2013stochastic}. 

The study of algorithms that compute competitive equilibria in Fisher markets was initiated by \citeauthor{devanur2002market}, who provided a polynomial-time method for solving these markets assuming linear utilities.
More recently, there have been efforts to study markets in dynamic settings \cite{cheung2019tracing, gao2021online, goktas2021minmax}, in which the goal is to either track the changing equilibrium of a changing market, or minimize some regret-like quantity for the market.
The models considered in these earlier works differ from ours as they do not have stochastic structure and do not invoke a dynamic solution concept. \deni{Add the work of Michael Jordan's group on stochastic EE.}

%% file: prelim.tex
\section{Preliminaries}\label{sec:prelim}

\paragraph{Notation} 
We use caligraphic uppercase letters to denote sets (e.g., $\calX$);
bold lowercase letters to denote vectors (e.g., $\price, \bm \pi$);
bold uppercase letters to denote matrices and vector-valued random variables (e.g., $\allocation$, $\bm \Gamma$)---which one should be clear from context;
lowercase letters to denote scalar quantities (e.g., $x, \gamma$);
and uppercase letters to denote scalar-valued random variables (e.g., $X, \Gamma$).
We denote the $i$th row vector of a matrix (e.g., $\allocation$) by the corresponding bold lowercase letter with subscript $i$ (e.g., $\allocation[\buyer])$. 
Similarly, we denote the $j$th entry of a vector (e.g., $\price$ or $\allocation[\buyer]$) by the corresponding Roman lowercase letter with subscript $j$ (e.g., $\price[\good]$ or $\allocation[\buyer][\good]$).
We denote functions by a letter:
e.g., $f$ if the function is scalar valued, and $\f$ if the function is vector valued.
We denote the vector of ones of size $\numbuyers$ by $\ones[\numbuyers]$.
We denote the set of integers $\left\{1, \hdots, n\right\}$ by $[n]$, the set of natural numbers by $\N$, the set of real numbers by $\R$. We denote the postive and strictly positive elements of a set by a $+$ and $++$ subscript respectively, e.g., $\R_+$ and $\R_{++}$. 
We denote the orthogonal projection operator onto a set $C$ by $\project[C]$, i.e., $\project[C](\x) = \argmin_{\y \in C} \left\|\x - \y \right\|^2$. We denote by $\simplex[n] = \{\x \in \R_+^n \mid \sum_{i = 1}^n x_i = 1\}$, and by $\simplex(A)$, the set of probability measures on the set $A$.

\deni{Should we stick to the outer, inner player language or should we instead embrace the leader-follower language since in the stochastic setting the problem "looks less" inner-outer?}
\amy{leader follower, but another time!}

A \mydef{stochastic Stackelberg game} $(\states, \outerset, \innerset, \initstates, \reward[\outer], \reward[\inner], \constr, \trans, \discount)$ is a two-player game played over an infinite discrete time horizon $\N_+$.
At each time-step $\iter \in \N_+$, the players, who we call the outer- (resp.{} inner-) players, encounter a new state $\state \in \states$, and choose an action to play from their continuous set of actions $\outerset \subset \R^\outerdim$ (resp.\ $\innerset \subset \R^\innerdim$).
Play initiates at a start state $\staterv[0]$ drawn from a distribution $\initstates: \states \to [0, 1]$.
At each state $\state \in \states$ the action $\outer \in \outerset$ chosen by the outer player determines the set of \mydef{feasible} actions $\{\inner \in \innerset \mid \constr(\state, \outer, \inner) \geq \zeros \}$ available to the inner player, where $\constr: \states \times \outerset \times \innerset \to \R^\numconstrs$.
After the outer and inner players both make their moves, they receive payoffs 
$\reward[\outer]: \states \times \outerset \times \innerset \to \R$ and $\reward[\inner]: \states \times \outerset \times \innerset \to \R$, respectively, and the game either ends with probability $1-\discount$, where $\discount \in (0,1)$ is called the \mydef{discount factor}, or transitions to a new state $\state^\prime \in \states$, according to a \mydef{transition} probability function $\trans: \states \times \states \times \outerset \times \innerset \to [0,1]$ s.t.\ $\trans(\state^\prime \mid \state, \outer, \inner) \in [0,1]$ denotes the probability of transitioning to state $\state^\prime \in \states$ from state $\state \in \states$ when action profile $(\outer, \inner) \in  \outerset \times \innerset$ is chosen by the players.


In this paper, we focus on \mydef{zero-sum} stochastic Stackelberg games 
$\initgame \doteq (\states, \outerset, \innerset, \initstates, \reward, \constr, \trans, \discount)$, 
in which the outer player's loss is the inner player's gain, i.e., $\reward[\outer] = -\reward[\inner]$. 
A zero-sum stochastic Stackelberg game reduces to zero-sum (simultaneous-move) stochastic game \cite{shapley1953stochastic} in the special case where $\constr(\state, \outer, \inner) \geq 0$, for all state-action tuples $(\state, \outer, \inner) \in \states \times \outerset \times \innerset$.
More generally, a policy profile $(\policy[\outer], \policy[\inner]) \in  \outerset^\states \times \innerset^\states$ is said to be \mydef{feasible} if $\constr(\state, \policy[\outer](\state), \policy[\inner](\state)) \geq 0$, for all states $\state \in \states$.
To simplify notation, we introduce a function $\constrs : \outerset^\states \times \innerset^\states \to \R^{\numstates \times \numconstrs}$ such that $\constrs(\policy[\outer], \policy[\inner]) = (\constr(\state, \policy[\outer](\state), \policy[\inner](\state)))_{\state \in \states}$, and define feasible policy profiles as those $(\policy[\outer], \policy[\inner]) \in  \outerset^\states \times \innerset^\states$ s.t.\ $\constrs(\policy[\outer], \policy[\inner]) \geq \zeros$.
From now on, we assume:

\begin{assumption}
\label{assum:main}
1.~For all states $\state \in \states$, the functions $\reward (\state, \cdot, \cdot), \constr(\state, \cdot, \cdot)$ are continuous in $(\outer, \inner) \in \outerset \times \innerset$,
and payoffs are bounded, i.e., $\left\| \reward \right\|_{\infty} \leq \rewardbound < \infty$, for some $\rewardbound \in \R_{+}$,
2.~$\outerset, \innerset$ are non-empty and compact, and for all $\state \in \states$ and $\outer \in \outerset$ there exists $\inner \in \innerset$ s.t.\ $\constr(\state, \outer, \inner) \geq \zeros$.%
\footnote{Note that this condition is weaker than Slater's condition; it simply ensures the feasible action sets are non-empty for the inner player at each state.}

\end{assumption} 

Given a zero-sum stochastic Stackelberg game $\initgame$, the \mydef{state-value function}, $\statevalue: \states \times \outerset^\states \times \innerset^\states \to \R$, and the \mydef{action-value function}, $\actionvalue: \states \times \outerset \times \innerset \times \outerset^\states \times \innerset^\states \to \R$, respectively, are defined as: \\
{\small\begin{align}
    \statevalue(\state; \policy[\outer], \policy[\inner] ) &= \mathop{\E^{\policy[\outer], \policy[\inner]}}_{\staterv[\iter+1] \sim \trans( \cdot \mid \staterv[\iter], \outerrv[\iter], \innerrv[\iter])}\left[\sum_{\iter = 0}^\infty 
    \discount^\iter \reward(\staterv[\iter], \outerrv[\iter], \innerrv[\iter]) \mid
    \staterv[0] = \state \right] \\ 
    \actionvalue(\state, \outer, \inner; \policy[\outer], \policy[\inner]) &= \mathop{\E^{\policy[\outer], \policy[\inner]}}_{\staterv[\iter+1] \sim \trans( \cdot \mid \staterv[\iter], \outerrv[\iter], \innerrv[\iter])}\left[\sum_{\iter = 0}^\infty 
    \discount^\iter \reward(\staterv[\iter], \outerrv[\iter], \innerrv[\iter]) \mid
    \staterv[0] = \state, \outerrv[0] = \outer, \innerrv[0] = \inner \right] 
\end{align}}

Again, to simplify notation, we write expectations conditional on $\outerrv[\iter] = \policy[\outer]( \staterv[\iter])$ and $\innerrv[\iter] = \policy[\inner](\staterv[\iter])$ as $\E^{\policy[\outer], \policy[\inner]}$, and denote the state- and action-value functions by $\statevalue[][{\policy[\outer]}][{ \policy[\inner]}](\state)$, and $\actionvalue[][{\policy[\outer]}][{ \policy[\inner]}](\state, \outer, \inner)$, respectively.
Additionally, we let $\statevalfuncs  = [-\nicefrac{\rewardbound}{1- \discount}, \nicefrac{\rewardbound}{1- \discount}]^{\states}$ be the space of all state-value functions of the form $\statevalue: \states \to [-\nicefrac{\rewardbound}{1- \discount}, \nicefrac{\rewardbound}{1- \discount}]$, and we let $\actionvalfuncs = [-\nicefrac{\rewardbound}{1- \discount}, \nicefrac{\rewardbound}{1- \discount}]^{\states \times \outerset \times \innerset}$ be the space of all action-value functions of the form $\actionvalue: \states \times \outerset \times \innerset \to [-\nicefrac{\rewardbound}{1- \discount}, \nicefrac{\rewardbound}{1- \discount}]$.
Note that by \Cref{assum:main} the range of the state- and action-value functions is $[-\nicefrac{\rewardbound}{1- \discount}, \nicefrac{\rewardbound}{1- \discount}]$.
The cumulative payoff function of the game $\cumul: \outerset^\states \times \innerset^\states \to \R$ is the total expected loss (resp.\ gain) of the outer (resp.\ inner) player, given by $\cumul (\policy[\outer], \policy[\inner]) = \mathop{\E}_{\state \sim \initstates (\state)} \left[ \statevalue[][{\policy[\outer]}][{\policy[\inner]}] (\state) \right]$.

The canonical solution concept for stochastic Stackelberg games is the \mydef{Stackelberg equilibrium (SE)}.
A feasible policy profile $(\policy[\outer]^*, \policy[\inner]^*) \in \outerset^\states \times \innerset^\states$ is said to be a Stackelberg equilibrium (SE) of a zero-sum stochastic Stackelberg game $\initgame$ iff \\
    $$\max_{\policy[\inner] \in \innerset^\states : \constrs( \policy[\outer]^*, \policy[\inner]) \geq \zeros}  \cumul \left( \policy[\outer]^*, \policy[\inner] \right) \leq \cumul \left( \policy[\outer]^*, \policy[\inner]^* \right)  \leq \min_{\policy[\outer] \in \outerset^\states} \max_{\policy[\inner] \in \innerset^\states : \constrs( \policy[\outer], \policy[\inner]) \geq \zeros}  \cumul \left( \policy[\outer], \policy[\inner] \right) \enspace .$$
Note the strength of this definition, as it requires the constraints $\constr(\state, \policy[\outer], \policy[\inner]) \geq \zeros$ to be satisfied at all states $\state \in \states$, not only states which are reached with strictly positive probability.
A SE is guaranteed to exist in zero-sum stochastic Stackelberg games, under \Cref{assum:main}, as a corollary of \citeauthor{goktas2021minmax}'s \cite{goktas2021minmax} Proposition B.2.; however, this existence result is non-constructive.%
\footnote{We note SE should technically be defined in terms of non-stationary policies; however, as we will show, stationary policies suffice, since SE exist in stationary policies.}

In this paper, we study a Markov perfect refinement of SE, which we call \mydef{recursive Stackelberg equilibrium} (\recSE{}).

\begin{definition}[Recursive Stackelberg Equilibrium (\recSE{})]\label{lemma:stackelberg_action_value}
Given $\initgame$, a policy profile $(\policy[\outer]^*, \policy[\inner]^*) \in \states^\outerset \times \states^\innerset$ is a \mydef{recursive Stackelberg equilibrium} (\recSE{}) iff, for all $\state \in \states$, it holds that:
{\small
$$
 \max_{\inner \in \innerset : \constr(\state, \policy[\outer]^*(\state), \inner) \geq \zeros}  \actionvalue[][{\policy[\outer]^*} ][{\policy[\inner]^*}](\state, \policy[\outer]^*(\state), \inner) \leq 
\actionvalue[][{\policy[\outer]^*}][ {\policy[\inner]^*}](\state, \policy[\outer]^*(\outer), \policy[\inner]^*(\inner))
\leq \min_{\outer \in \outerset} \max_{\inner \in \innerset : \constr(\state, \outer, \inner) \geq \zeros} \actionvalue[][{\policy[\outer]^*}][ {\policy[\inner]^*}](\state, \outer, \inner).
$$
}
\end{definition}

Equivalently, a policy profile $(\policy[\outer]^*, \policy[\inner]^*)$ is a \recSE{} if $(\policy[\outer]^*(\state), \policy[\inner]^*(\state))$ is a SE with value $\statevalue[][{\policy[\outer]^*}][ {\policy[\inner]^*}](\state)$ at each state $\state \in \states$:
i.e., $\statevalue[][{\policy[\outer]^*}][ {\policy[\inner]^*}](\state) = \min_{\outer \in \outerset} \max_{\inner \in \innerset : \constr(\state, \outer, \inner) \geq \zeros} \actionvalue[][{\policy[\outer]^*}][ {\policy[\inner]^*}](\state, \outer, \inner)$, for all 
$\state \in \states$.

\paragraph{Mathematical Preliminaries}
A probability measure $q_1 \in \simplex(\states)$ \mydef{convex stochastically dominates (CSD)} $q_2 \in \simplex(\states)$ if $\int_\states \statevalue(s) q_1(s) ds \geq \int_\states \statevalue(s) q_2(s) ds$ for all continuous, bounded, and convex functions $\statevalue$ on $S$.
A transition function $\trans$ is termed \mydef{CSD convex} in $\outer$ if, for all $\lambda \in (0,1)$, $\inner \in \innerset$ and any $(\state^\prime, \outer^\prime), (\state^\dagger, \outer^\dagger) \in \states \times \outerset$, with $(\state, \outer) = \lambda (\state^\prime, \outer^\prime) + (1-\lambda) (\state^\dagger, \outer^\dagger)$, it holds that $\lambda \trans(\cdot \mid \state^\prime, \outer^\prime, \inner) + (1-\lambda) \trans(\cdot \mid \state^\dagger, \outer^\dagger, \inner)$ CSD $\trans(\cdot \mid \state, \outer, \inner)$.
A transition function $\trans$ is termed \mydef{CSD concave} in $\inner$ if, for all $\lambda \in (0,1)$ and any $(\state^\prime, \inner^\prime), (\state^\dagger, \inner^\dagger) \in  \states \times \outerset \times \innerset$, with $(\state, \inner) = \lambda (\state^\prime,  \inner^\prime) + (1-\lambda) (\state^\dagger, \inner^\dagger)$, it holds that $\trans(\cdot \mid \state, \outer, \inner)$ CSD $\lambda \trans(\cdot \mid \state^\prime, \outer, \inner^\prime) + (1-\lambda) \trans(\cdot \mid \state^\dagger, \outer, \inner^\dagger)$.
A mapping $L : \calA \to \calB$ is said to be a \mydef{contraction mapping} (resp.{} \mydef{non-expansion}) w.r.t.\ norm $\left\| \cdot \right\|$ iff for all $\x, \y \in \calA$, and for $k \in [0,1)$ (resp.{} $k = 1$) such that $\left\| L(\x) - L(\y) \right\| \leq k \left\| \x - \y \right\|$.
The \mydef{min-max operator} $\min_{\outer \in \outerset} \max_{\inner \in \innerset}: \R^{\outerset \times \innerset} \to \R$ w.r.t. to sets $\outerset, \innerset$ takes as input a function $\obj: \outerset \times \innerset \to \R$ and outputs $\min_{\outer \in \outerset} \max_{\inner \in \innerset} \obj(\outer, \inner)$.
The \mydef{generalized min-max operator} $\min_{\outer \in \outerset} \max_{\inner \in \innerset: \constr(\outer, \inner) \geq \zeros}: \R^{\outerset \times \innerset} \to \R$ w.r.t. to sets $\outerset, \innerset$ and the function $\constr: \outerset \times \innerset \to \R$ takes as input a function $\obj: \outerset \times \innerset \to \R$ and outputs $\min_{\outer \in \outerset} \max_{\inner \in \innerset: \constr(\outer, \inner) \geq \zeros} \obj(\outer, \inner)$.

%% file: properties.tex
\section{Properties of Recursive Stackelberg equilibrium}
\label{sec:properties}

In this section, we show that a \recSE{} exists in all zero-sum stochastic Stackelberg games.%
\footnote{All omitted results and proofs can be found in the appendix.}
To do so, we first associate an operator $\bellopt: \statevalfuncs \to \statevalfuncs$ with any zero-sum stochastic Stackelberg game $\initgame$, the fixed points of which we show satisfy \Cref{lemma:stackelberg_action_value}, and hence are necessary and sufficient to
characterize \deni{Is necessary and sufficient redundant here? I don't feel it is?}\amy{i do think it is redundant. i think that is exactly what characterization means. but i am also fine with saying this twice}
the value function associated with a \recSE{} of $\initgame$.
We then show that this operator is a contraction mapping, thereby establishing the existence of such a fixed point.
Together these results generalize \citeauthor{shapley1953stochastic}'s theorem on the existence of Markov perfect Nash equilibria in zero-sum stochastic games \cite{shapley1953stochastic}.

Given a zero-sum stochastic Stackelberg game $\initgame$, define $\bellopt: \statevalfuncs \to \statevalfuncs$ as the operator $\left(\bellopt \statevalue[] \right) (\state) = \min_{\outer \in \outerset} \max_{\inner \in \innerset : \constr(\state, \outer, \inner)\geq \zeros} \mathop{\E}_{\staterv^\prime \sim \trans(\cdot \mid \state, \outer, \inner)} \left[ \reward(\state, \outer, \inner) +  \discount \statevalue[](\staterv^\prime)  \right]$. 
We call the set of $\numstates$ equations $\left(\bellopt \statevalue[] \right) (\state) =  \statevalue[](\state)$ with $\statevalue[]$ unknown the \mydef{Bellman equations} for $\initgame$.
The next theorem shows that these Bellman equations
characterize the state-value function associated with a \recSE.
Consequently, the solution to a zero-sum stochastic Stackelberg game can be described by its Bellman equations.

 

\begin{restatable}{theorem}{thmfpbelloptisstackelberg}
\label{thm:fp_bellopt_is_stackelberg}
Given $\initgame$, the policy profile $(\policy[\outer]^*, \policy[\inner]^*)$ is a \recSE{} iff it induces a value function which is a fixed point of $\bellopt$: i.e., $(\policy[\outer]^*, \policy[\inner]^*)$ is a \recSE{} iff, for all $\state \in \states,$ $\left(\bellopt \statevalue[][{\policy[\outer]^*}][{\policy[\inner]^*}] \right) (\state) = \statevalue[][{\policy[\outer]^*}][{\policy[\inner]^*}](\state)$.
\end{restatable}
The following technical lemma is crucial to proving that $\bellopt$ is a contraction mapping.
It tells us that the generalized min-max operator is non-expansive; in other words, the generalized min-max operator is 1-Lipschitz w.r.t.\ the sup-norm.

\begin{restatable}{lemma}{lemmaminmaxlipschitz}
\label{lemma:minmax_lipschitz}
Suppose that $f, h: \outerset \times \innerset \to \R$, $\constr: \outerset \times \innerset \to \R^\numconstrs$ are continuous functions, and $\outerset, \innerset$ are compact sets.
Then
$\left| \min_{\outer \in \outerset} \max_{\inner \in \innerset : \constr(\outer, \inner) \geq \zeros} f(\outer, \inner) - \min_{\outer \in \outerset} \max_{\inner \in \innerset : \constr(\outer, \inner) \geq \zeros} h(\outer, \inner) \right|$ \\
$\leq \max_{(\outer, \inner) \in \outerset \times \innerset } \left| f(\outer, \inner) -  h(\outer, \inner) \right|$.
\end{restatable}


\begin{restatable}{theorem}{thmcontractionmapping}\label{thm:contraction_mapping}
Consider the operator $\bellopt$ associated with a zero-sum stochastic Stackelberg game $\initgame$. Under \Cref{assum:main}, $\bellopt$ is a contraction mapping w.r.t.\ to the sup norm $\left\| . \right\|_{\infty}$ with 
constant $\gamma$.
\end{restatable}

\begin{proof}
We will show that $\bellopt$ is a contraction mapping.
The result then follows by Banach's fixed point theorem \cite{banach1922operations}.
Let $\statevalue[], {\statevalue[]}^\prime \in \statevalfuncs$ be any two state value functions, and let $\actionvalue, {\actionvalue}^\prime \in \actionvalfuncs$ be the associated action-value functions, respectively.
Then, it holds that
\begin{align}
    &\left\| \bellopt \statevalue[] - \bellopt {\statevalue[]}^\prime \right\|_{\infty} \nonumber \\
    &= \max_{\state \in \states} \left|\min_{\outer \in \outerset} \max_{\inner \in \innerset : \constr(\state, \outer, \inner) \geq  \zeros} \actionvalue(\state, \outer, \inner) - \min_{\outer \in \outerset} \max_{\inner \in \innerset : \constr(\state, \outer, \inner) \geq  \zeros} {\actionvalue}^{\prime}(\state, \outer, \inner) \right| \nonumber \\
    &\leq \max_{\state \in \states} \max_{(\outer, \inner) \in \outerset \times \innerset} \left| \actionvalue(\state, \outer, \inner) - {\actionvalue}^{\prime}(\state, \outer, \inner) \right| \label{eq:dist_min_maxes} \\
    &= \max_{\state \in \states} \max_{(\outer, \inner) \in \outerset \times \innerset} \left|\mathop{\E}_{\staterv^\prime \sim \trans(\cdot \mid \state, \outer, \inner)} \left[ \reward(\state, \outer, \inner) +  \discount \statevalue[](\staterv^\prime)  \right] - \mathop{\E}_{\staterv^\prime \sim \trans(\cdot \mid \state, \outer, \inner)} \left[ \reward(\state, \outer, \inner) +  \discount {\statevalue[]}^\prime(\staterv^\prime)  \right] \right| \nonumber \\
    &=\discount \max_{\state \in \states} \max_{(\outer, \inner) \in \outerset \times \innerset} \left|\mathop{\E}_{\staterv^\prime \sim \trans(\cdot \mid \state, \outer, \inner)} \left[ \statevalue[](\staterv^\prime)  -  {\statevalue[]}^\prime(\staterv^\prime)  \right] \right| \nonumber \\
    &\leq\discount \max_{\state \in \states} \max_{(\outer, \inner) \in \outerset \times \innerset} \max_{\state^\prime \in \states} \left| \statevalue[](\state^\prime)  -  {\statevalue[]}^\prime(\state^\prime) \right| \label{eq:exp_less_than_max} \\
    &= \discount \max_{\state \in \states} \max_{(\outer, \inner) \in \outerset \times \innerset} \left| \statevalue[](\state)  -  {\statevalue[]}^\prime(\state) \right| \nonumber \\
    &= \discount \left\|\statevalue[] - {\statevalue[]}^\prime \right\|_\infty \enspace , \nonumber
\end{align}

\noindent
where \Cref{eq:dist_min_maxes} follows by \Cref{lemma:minmax_lipschitz} and \Cref{eq:exp_less_than_max} follows from the fact that the expectation of a random variable is less than or equal to its maximum value.
\end{proof}

Given an initial state-value function $\statevalue[0] \in \statevalfuncs$, we define the \mydef{value iteration} process  as $\statevalue[\iter+1] = \bellopt \statevalue[\iter]$, for all $\iter \in \N_+$ (\Cref{alg:value_iter}).
One way to interpret $\statevalue[\iter]$ is as the function that returns the value $\statevalue[\iter](\state)$ of each state $\state \in \states$ in the $\iter$-stage zero-sum stochastic Stackelberg game starting at the last stage $\iter$ and proceeding backwards to stage $0$, with terminal payoffs given by $\statevalue[0]$.
The following theorem, which is a consequence of Theorems~\ref{thm:fp_bellopt_is_stackelberg} and~\ref{thm:contraction_mapping}, and the Banach fixed point theorem \cite{banach1922operations}, not only proves the existence of a \recSE{}, but further provides us with a means of computing a \recSE{} via value iteration. 

\begin{restatable}{theorem}{thmexistencestackelberg}
\label{thm:existence_stackelberg}
Consider a zero-sum stochastic Stackelberg game $\initgame$.
Under \Cref{assum:main}, $\initgame$ has
a unique
value function $\statevalue[][{\policy[\outer]^*}][{\policy[\inner]^*}]$ associated with all \recSE{} $(\policy[\outer]^*, \policy[\inner]^*)$,
which can be computed by iteratively applying $\bellopt$ to
any initial state-value function $\statevalue[0] \in \statevalfuncs$: i.e., $\lim_{\iter \to \infty} \statevalue[\iter] = \statevalue[][{\policy[\outer]^*}][{\policy[\inner]^*}]$.
\end{restatable}

\begin{remark}
Unlike Shapley's existence theorem for recursive Nash equilibria in zero-sum stochastic games, \Cref{thm:existence_stackelberg} does not require the payoff function to be convex-concave.
The only conditions needed are continuity of the payoffs and constraints, and bounded payoffs.
This makes the \recSE{} a potentially useful solution concept, even for non-convex-non-concave stochastic games. 
\end{remark}


Since a \recSE{} is guaranteed to exist, and is by definition independent of the initial state distribution, we can infer that the \recSE{} of any zero-sum stochastic Stackelberg game $\initgame = (\states, \outerset, \innerset, \initstates, \reward, \constr, \trans, \discount)$ is independent of the initial state distribution $\initstates$.
Hence,
in the remainder of the paper, we denote zero-sum stochastic Stackelberg games by $\game \doteq (\states, \outerset, \innerset, \reward, \constr, \trans, \discount)$. 

\Cref{thm:existence_stackelberg} tells us that value iteration 
converges to the value function associated with a \recSE.
Additionally, under \Cref{assum:main}, \recSE{} is computable in (weakly) polynomial time.%
\footnote{This convergence is only weakly polynomial time, because the computation of the generalized min-max operator applied to an arbitrary continuous function is an NP-hard problem; it is at least as hard as non-convex optimization.
If, however, we restrict attention to convex-concave stochastic Stackelberg games, then Stackelberg equilibria are computable in polynomial time.}

\begin{restatable}[Convergence of Value Iteration]{theorem}{thmvalueiter}\label{thm:value_iter}
Suppose value iteration is run on input $\game$.
Let $(\policy[\outer]^*, \policy[\inner]^*)$ be \recSE{} of $\game$ with value function $\statevalue[][{\policy[\outer]^*}][{\policy[\inner]^*}]$.
Under \Cref{assum:main}, if we initialize $\statevalue[0] (\state) = 0$, for all $\state \in \states$, then for $k \geq \frac{1}{1 - \discount}  \log \frac{ \rewardbound}{\epsilon(1 - \discount)}$, it holds that
    $\statevalue[k](\state)  - \statevalue[][{\policy[\outer]^*}][{\policy[\inner]^*}](\state) \leq \epsilon$. 
\end{restatable}

%% file: subdifferential.tex
\section{Subdifferential Envelope Theorems and Optimality Conditions for Recursive Stackelberg Equilibrium}\label{sec:benveniste}

In this section, we derive optimality conditions for recursive Stackelberg equilibria. 
In particular, we provide necessary conditions for a policy profile to be a \recSE{} of any zero-sum stochastic Stackelberg game, and show that under additional convexity assumptions, these conditions are also sufficient. 

The Benveniste-Scheinkman theorem characterizes the derivative of the optimal value function associated with a recursive optimization problem
w.r.t.\ its parameters, when it is differentiable \cite{Benveniste1979value}.
Our proofs of the necessary and sufficient optimality conditions rely on a novel subdifferential generalization (\Cref{thm:subdiff_bellman}, \Cref{sec_app:benveniste}) of this theorem, which applies even when the optimal value function is not differentiable.
A consequence of our subdifferential version of the Benveniste-Scheinkman theorem is that we can easily derive the first-order necessary 
conditions for a policy profile to be a \recSE{} of any zero-sum stochastic Stackelberg game $\game$ satisfying \Cref{assum:main}, under standard regularity conditions.

\begin{restatable}{theorem}{thmnecessaryconditions}
\label{thm:necessary_conditions}
Consider a zero-sum stochastic Stackelberg game $\game$, where $\outerset = \{\outer \in \R^\outerdim \mid \outerconstr[1](\outer)\leq 0, \hdots, \outerconstr[\outernumconstrs](\outer) \leq 0 \}$ and $\innerset = \{\inner \in \R^\innerdim \mid 
\innerconstr[1](\inner)\geq 0, \hdots, \innerconstr[\innernumconstrs](\inner)  \geq 0 \}$ are convex.
Let $\lang[\state, \outer](\inner, \langmult) =  \reward(\state, \outer, \inner) + \discount \mathop{\E}_{\staterv^\prime \sim \trans(\cdot \mid \state, \outer, \inner)} \left[\statevalue(\staterv^\prime, \outer) \right] +  \sum_{\numconstr = 1}^\numconstrs \langmult[\numconstr] \constr[\numconstr](\state, \outer, \inner)$ where $\bellopt \statevalue = \statevalue$.

Suppose that  \Cref{assum:main} holds, and that 
1.~for all $\state \in \states$,
$\max_{\inner \in \innerset: \constr(\state, \outer, \inner) \geq \zeros} \left\{ \reward(\state, \outer, \inner) + \discount \mathop{\E}_{\staterv^\prime \sim \trans(\cdot \mid \state, \outer, \inner)} \left[\statevalue(\staterv^\prime, \outer) \right] \right\}$
is
concave in $\outer$, 
2.~$\grad[\outer] \reward(\state, \outer, \inner), \grad[\outer] \constr[1](\state, \outer, \inner), \ldots, \grad[\outer] \constr[\numconstrs](\state, \outer, \inner)$, $\grad[\inner] \reward(\state, \outer, \inner), \grad[\inner] \constr[1](\state, \outer, \inner), \ldots, \grad[\inner] \constr[\numconstrs](\state, \outer, \inner)$ exist, for all $\state \in \states, \outer \in \outerset, \inner \in \innerset$,
4.~$\trans(\state^\prime \mid \state, \outer, \inner)$ is continuous
and differentiable in $(\outer, \inner)$, and
5.~Slater's condition holds, i.e., $\forall \state \in \states, \outer \in \outerset, \exists \widehat{\inner} \in \innerset$ s.t.\ $\constr[\numconstr](\state, \outer, \widehat{\inner}) > 0$, for all $\numconstr = 1, \ldots, \numconstrs$ and $\innerconstr[j](\widehat{\inner})  > 0$, for all $j = 1, \hdots, \innernumconstrs$, and $\exists \outer \in \R^\outerdim$ s.t. $\outerconstr[\numconstr](\outer) < 0$ for all $\numconstr =1 \hdots, \outernumconstrs$.
Then, there exists $\bm{\mu}^*: \states \to \R_+^\outernumconstrs$, $\langmult^* : \states \times \outerset \to \R_+^\numconstrs$, and $\bm{\nu}^* : \states \times \outerset \to \R_+^\innernumconstrs$ s.t.\ a policy profile $(\policy[\outer]^*, \policy[\inner]^*) \in \outerset^\states \times \innerset^\states$ is a \recSE{} of $\game$ only if it satisfies the following conditions, for all $\state \in \states$:
\begin{align}
    &\grad[\outer] \lang[\state, {\policy[\outer]^*(\state)}](\policy[\inner]^*(\state), \langmult^*(\state, \policy[\outer]^*(\state))) + \sum_{\numconstr = 1}^\outernumconstrs \mu^*_\numconstr(\state) \grad[\outer]\outerconstr[\numconstr](\policy[\outer]^*(\state)) = 0\\
    &\grad[\inner] \lang[\state, {\policy[\outer]^*( \state)}](\policy[\inner]^*(\state), \langmult^*(\state, \policy[\outer]^*(\state))) + \sum_{\numconstr = 1}^\innernumconstrs \nu^*_\numconstr(\state, \policy[\outer]^*(\state)) \grad[\outer]\innerconstr[\numconstr](\policy[\inner]^*(\state)) = 0
\end{align}
\begin{align}
    &\mu^*_\numconstr(\state)\outerconstr[\numconstr](\policy[\outer]^*(\state)) = 0 &\outerconstr[\numconstr](\policy[\outer]^*(\state)) \leq 0 && \forall \numconstr \in [\outernumconstrs]\\
    &\constr[\numconstr](\state, \policy[\outer]^*(\state), \policy[\inner]^*(\state)) \geq 0 &\langmult[\numconstr]^*(\state, \policy[\outer]^*(\state)) \constr[\numconstr](\state, \policy[\outer]^*(\state), \policy[\inner]^*(\state)) = 0  && \forall \numconstr \in [\numconstrs]\\
    &\nu^*_\numconstr(\state, \policy[\outer]^*(\state)) \grad[\outer]\innerconstr[\numconstr](\policy[\inner]^*(\state)) = 0  &\innerconstr[\numconstr](\policy[\outer]^*(\state)) \geq 0 && \forall \numconstr \in [\innernumconstrs]
\end{align}
\end{restatable}

Under the conditions of \Cref{thm:necessary_conditions}, if we additionally assume that
for all $\state \in \states$ and $\outer \in \outerset$, both $\reward(\state, \outer, \inner)$ and $\constr[1](\state, \outer, \inner), \hdots, \constr[\numconstrs](\state, \outer, \inner)$
are concave in $\inner$, and $\trans(\state^\prime \mid \state, \outer, \inner)$ is continuous, CSD concave in $ \inner$, and differentiable in $(\outer, \inner)$, 
\Crefrange{eq:optimality_inner_1}{eq:optimality_inner_5} become necessary \emph{and sufficient} optimality conditions.
For completeness, the reader can find the necessary and sufficient optimality conditions for convex-concave stochastic Stackelberg games under standard regularity conditions in \Cref{thm:necessary_sufficient_conditions} (\Cref{sec_app:benveniste}).
The proof follows exactly as that of \Cref{thm:fp_bellopt_is_stackelberg}.


%% file: fisher.tex
\section{Recursive Market Equilibrium}
\label{sec:fisher}

We now introduce 
an application of zero-sum stochastic Stackelberg games, which generalizes a well known market model, the Fisher market \cite{brainard2000compute}, to a dynamic setting in which buyers not only participate in markets across time, but their wealth persists. 
A \mydef{(one-shot) Fisher market} consists of $\numbuyers$ buyers and $\numgoods$ divisible goods \cite{brainard2000compute}.
Each buyer $\buyer \in \buyers$ is endowed with a budget $\budget[\buyer] \in \budgetspace[\buyer] \subseteq \mathbb{R}_{+}$ and a utility function $\util[\buyer]: \mathbb{R}_{+}^{\numgoods} \times \typespace[\buyer] \to \mathbb{R}$, which is parameterized by a type $\type[\buyer] \in \typespace[\buyer]$ that defines a preference relation over the consumption space $\R^\numgoods_+$.
Each good is characterized by a supply $\supply[\good] \in \supplyspace[\good] \subset \R_+$.

An instance of a Fisher market is then a tuple $\calM \doteq (\numbuyers, \numgoods, \util, \type, \budget, \supply)$, where $\util = \left\{\util[1], \hdots, \util[\numbuyers] \right\}$ is a set of utility functions, one per buyer,
$\budget \in \R_{+}^{\numbuyers}$ is the vector of buyer budgets, and $\supply \in \R_{+}^{\numgoods}$ is the vector of supplies.
When clear from context, we simply denote $\calM$
by $(\type, \budget, \supply)$.

Given a Fisher market $(\type, \budget, \supply)$, an \mydef{allocation} $\allocation = \left(\allocation[1], \hdots, \allocation[\numbuyers] \right)^T \in \R_+^{\numbuyers \times \numgoods}$ is a map from goods to buyers, represented as a matrix, s.t. $\allocation[\buyer][\good] \ge 0$ denotes the amount of good $\good \in \goods$ allocated to buyer $\buyer \in \buyers$. Goods are assigned \mydef{prices} $\price = \left(\price[1], \hdots, \price[\numgoods] \right)^T \in \mathbb{R}_+^{\numgoods}$. A tuple $(\allocation^*, \price^*)$ is said to be a \mydef{competitive equilibrium (CE)} \cite{arrow-debreu, walras} if 
1.~buyers are utility maximizing, constrained by their budget, i.e., $\forall \buyer \in \buyers, \allocation[\buyer]^* \in \argmax_{\allocation[ ] : \allocation[ ] \cdot \price^* \leq \budget[\buyer]} \util[\buyer](\allocation[ ], \type[\buyer])$;
and 2.~the market clears, i.e., $\forall \good \in \goods,  \price[\good]^* > 0 \Rightarrow \sum_{\buyer \in \buyers} \allocation[\buyer][\good]^* = \supply[\good]$ and $\price[\good]^* = 0 \Rightarrow\sum_{\buyer \in \buyers} \allocation[\buyer][\good]^* \leq \supply[\good]$.

A \mydef{stochastic Fisher market with savings} is a dynamic market in which each state corresponds to a static Fisher market: i.e., each state $\state \in \states = \typespace \times \budgetspace \times \supplyspace$ is characterized by a tuple $\state \doteq (\type, \budget, \supply)$.
The market is initialized at state $\staterv[0] \sim \initstates$, and for each state $\staterv[\iter] \in \states$ encountered at time step $\iter \in \N_+$, the market determines the prices $\price[][\iter]$ of the goods, while the buyers choose their allocations $\allocation[][][\iter]$ and set aside some \mydef{savings} $\saving[\buyer][\iter] \in [0, \budget[\buyer]]$ to potentially spend at some future state.
Once allocations, savings, and prices have been determined, the market terminates with probability $1-\discount$, or it transitions to a new state $\staterv[\iter + 1] = \state^\prime$ with probability $\discount \trans(\state^\prime \mid \state, \saving)$, depending on the buyers' saving decisions.%
\footnote{In our model, which is consistent with the literature~\cite{romer2012advanced} 1.~prices do not determine the next state since market prices are set by a ``fictional auctioneer,'' not an actual market participant; 2.~allocations do not determine the next state.
Only savings, which are forward-looking decisions, affect future states---budgets, specifically.}
We denote a stochastic Fisher market by $\fishermkt[0] \doteq (\numbuyers, \numgoods, \util, \states, \initstates, \trans, \discount)$.

Given a stochastic Fisher market with savings $\fishermkt[0]$
a \mydef{recursive competitive equilibrium (recCE)} \cite{mehra1977recursive} is a tuple $(\allocation^*, \saving^*, \price^*) \in \R_+^{\numbuyers \times \numgoods \times \states} \times \R_+^{\numbuyers \times \states} \times \R_+^{\numgoods \times \states}$, which consists of stationary \mydef{allocation}, \mydef{savings}, and \mydef{pricing policies} s.t.\ 1) the buyers are expected utility maximizers, constrained by their savings and spending constraints, i.e., for all buyers $\buyer \in \buyers$, $(\allocation[\buyer]^*, \saving[\buyer]^*)$ is an optimal policy that, for all states $\state \doteq (\type, \budget, \supply) \in \states$, solves the \mydef{consumption-savings problem}, defined by the following Bellman equations: for all $\state \in \states$, $\budgetval[\buyer](\state) =$
$$
    \max_{(\allocation[\buyer], \saving[\buyer]) \in \R^{\numgoods + 1}_+: \allocation[\buyer] \cdot \price^*(\state) + \saving[\buyer] \leq \budget[\buyer]} \left\{ \util[\buyer]\left(\allocation[\buyer], \type[\buyer]\right) 
    + \discount \mathop{\E}_{(\type^\prime, \budget^\prime, \supply^\prime) \sim \trans(\cdot \mid \state, (\allocation[\buyer], \allocation_{-\buyer}^*(\state)), (\saving[\buyer], \saving_{-\buyer}^*(\state)))} \left[ \budgetval[\buyer](\type^\prime, \budget^\prime + \saving[\buyer], \supply^\prime) \right] \right\},
$$ 
where $\allocation^*_{-\buyer}$, $\saving^*_{-\buyer}$ denote the 
allocation and saving policies excluding buyer $\buyer$; and
2) the market clears in each state so that unallocated goods in each state are priced at 0, i.e., for all $\good \in \goods$ and $\state \in \states$,
$    
        \price[\good]^*(\state) > 0 \implies \sum_{\buyer \in \buyers} \allocation[\buyer][\good]^*(\state) = \supply[\good]
$ and
$
        \price[\good]^*(\state) \geq 0 \implies \sum_{\buyer \in \buyers} \allocation[\buyer][\good]^*(\state) \leq \supply[\good]
$.
A recCE is Markov perfect, as it is a CE regardless of initial  state,%
\footnote{Just as any stochastic game can be (non-compactly) represented as a one-shot game, any stochastic Fisher market with savings can be represented as a one-shot Fisher market comprising the same buyers (with utility functions given by their discounted cumulative expected utility) 
\samy{}{and the same goods\deni{saying the same goods is confusing because in the new market you think of the goods at each time step as being seperate good, I like the previous version better.}, enhanced with time-stamps.}
It has been shown that competitive equilibria exist in such markets\cite{prescott1972note}.
In particular, a recCE is a CE in this market; the time stamps ensure that the market clears at every time step, as required.}
i.e., buyers are allocated expected discounted cumulative utility-maximizing goods regardless of initial state, and the aggregate demand for each good is equal to its aggregate supply \emph{at all states}.
\amy{why does the qualifier ``starting from any state'', or ``regardless of initial conditions'' apply only to the first half of this sentence and not the second?}\deni{Because the market clearance condition is more robust and holds in every time period. this is because the seller is myopic at recCE; only the buyers are long-term optimizers.}


The following theorem establishes that we can build a zero-sum stochastic Stackelberg game whose \recSE{} correspond to the set of recCE of any stochastic Fisher market with savings.
Since \recSE{} are guaranteed to exist, recCE are also guaranteed to exist.
As \recSE, and hence recCE, are independent of the initial market state, we 
denote stochastic Fisher markets with savings by $\fishermkt \doteq (\numbuyers, \numgoods, \util, \states, \trans, \discount)$.

\begin{restatable}{theorem}{thmfishermarketrecursiveeqm}
\label{thm:fisher_market_recursive_eqm}
A stochastic Fisher market with savings $\fishermkt$
in which $\util$ is a set of continuous and homogeneous utility functions and the transition function is continuous in $\saving[\buyer]$ has at least one recCE.
Further, the \recSE{}
that solves the following Bellman equations corresponds to a recCE of $\fishermkt$: 
\begin{align} \label{eq:fisher_bellman}
    \forall \state \in \states, \statevalue[][](\state) = \min_{\price \in \R_+^\numgoods} \max_{(\allocation, \saving) \in \R_+^{\numbuyers \times (\numgoods + 1)} : \allocation \price + \saving \leq \budget} \sum_{\good \in \goods} \supply[\good] \price[\good] 
    + \sum_{\buyer \in \buyers} \left(\budget[\buyer] - \saving[\buyer]\right) \log(\util[\buyer](\allocation[\buyer], \type[\buyer])) \notag\\ 
    + \discount \mathop{\E}_{(\type^\prime, \budget^\prime, \supply^\prime) \sim \trans(\cdot \mid \state, 
    \saving)} \left[\statevalue[][](\type^\prime, \budget^\prime + \saving, \supply^\prime) \right]
\end{align} 
\end{restatable}

\amy{the thm states that every recSE is a recCE. are you saying that there may be other recCEs that we do not model/cannot discover?}\deni{Yes, that is correct}

\amy{\begin{remark}
This theorem states that every recSE is a recCE.
But every recCE is not a recSE. \amy{add example!}\deni{The second part of this statement is not necesarrily true, I do not have an example. Happy to think about it another time but this is too time consuming for right now.}
\end{remark}}
 
\begin{remark}
This result cannot be obtained by modifying the Lagrangian formulation,
i.e., the simultaneous-move game form, of the Eisenberg-Gale program, because the inner maximization problem is convex-non-concave, but recursive Nash equilibria 
are guaranteed to exist in zero-sum stochastic games only under the assumption of convex-concave payoffs \cite{jaskiewicz2018non}.
\end{remark}

%% file: experiments.tex
\section{Experiments}
\label{sec:expts}

The zero-sum stochastic Stackelberg game associated with a stochastic Fisher market, can, in theory, be solved via value iteration (\Cref{alg:value_iter_fisher}), assuming one can perform a one-step backup operation of the corresponding Bellman equations (Equation~\ref{eq:fisher_bellman}), i.e., solve the min-max optimization problem given in Line 4 of \Cref{alg:value_iter_fisher}.
The $\sum_{\buyer \in \buyers} \left(\budget[\buyer] - \saving[\buyer]\right) \log(\util[\buyer](\allocation[\buyer], \type[\buyer]))$ term renders the objective function in this min-max optimization problem convex-\emph{non}-concave.
As a result, known (first-order) methods, e.g., nested gradient descent ascent (GDA) \cite{goktas2021minmax}, are not guaranteed to find a globally optimal solution;
and without a globally optimal solution to the inner optimization problem, our value iteration convergence guarantees (\Cref{thm:existence_stackelberg}) do not hold.
Still, gradient methods have been observed to escape local solutions in many non-convex optimization problems (e.g., \cite{jin2017escape, du2017gradient}), leading us to investigate the extent to which \Cref{alg:value_iter_fisher} can solve stochastic Fisher markets with savings using nested GDA \cite{goktas2021minmax} as a subroutine to solve the inner optimization problem.

\begin{algorithm}[H]
\caption{Value Iteration for Stochastic Fisher Market}
\label{alg:value_iter_fisher}
\begin{algorithmic}[1]
\State Initialize $\statevalue[0]$ arbitrarily, e.g. $\statevalue[0]=\zeros$
\For{$k = 1, \hdots, \numiters_{\statevalue}$}
        \State For all states $\state \in \states$, $\displaystyle
        \statevalue[k+1] (\type, \budget, \supply) = $ \\ $\min\limits_{\price \geq \zeros} \max\limits_{(\allocation, \saving) \geq \zeros: \allocation \price + \saving \leq \budget} \left\{
        \sum\limits_{\good \in \goods} \supply[\good] \price[\good] + \sum\limits_{\buyer \in \buyers} \left(\budget[\buyer] - \saving[\buyer]\right) \log(\util[\buyer](\allocation[\buyer], \type[\buyer])) 
         + \discount \mathop{\E} \left[\statevalue[k][](\type^\prime, \budget^\prime + \saving, \supply^\prime) \right] \right\}$
\EndFor
\end{algorithmic}
\end{algorithm}

We attempted to compute the recursive competitive equilibria of three different classes of stochastic Fisher markets with savings.%
\footnote{Our code can be found \href{\rawcoderepo}{here}, and details of our experimental setup can be found in \Cref{sec_app:experiments}.}
Specifically, we created markets with three classes of utility functions, each of which endowed the state-value function with different smoothness properties.
Let $\type[\buyer] \in \R^\numgoods$ be a vector of parameters, i.e., a \mydef{type}, that describes the utility function of buyer $\buyer \in \buyers$.
We considered the following (standard) utility function classes:
1.~\mydef{linear}: $\util[\buyer](\allocation[\buyer]) = \sum_{\good \in \goods} \type[\buyer][\good] \allocation[\buyer][\good]$; 2.~\mydef{Cobb-Douglas}:  $\util[\buyer](\allocation[\buyer]) = \prod_{\good \in \goods} \allocation[\buyer][\good]^{\type[\buyer][\good]}$; and 3.~\mydef{Leontief}:  $\util[\buyer](\allocation[\buyer]) = \min_{\good \in \goods} \left\{ \frac{\allocation[\buyer][\good]}{\type[\buyer][\good]}\right\}$. 

We ran two different experiments.
First, we modeled a small stochastic Fisher market with savings \emph{without\/} interest rates.
In this setting, buyers' budgets, which are initialized at the start of the game, persist across states, and are replenished by a constant amount with each state transition.
Thus, the buyers' budgets from one state to the next are deterministic.

Second, we modeled a larger stochastic Fisher market with savings and probabilistic interest rates.
In this model, although buyers' savings persist across states, they are nondeterministic, as they increase or decrease based on the random movements of an interest rate with each state transition. 
More specifically, we chose five different equiprobable interest rates (0.9, 1.0, 1.1, 1.2, and 1.5) to provide buyers with more incentive to save as compared to the model without interest rates.

Since budgets are a part of the state space in stochastic Fisher markets, the state space is continuous; so we attempted to estimate the value function using linear regression, with a uniform sample of state and min-max value pairs (e.g., \cite{boyan1994generalization}).
To compute the min-max value at each sampled state, 
we used nested GDA \cite{goktas2021minmax}, which, in this application, repeatedly takes a step of gradient descent on the prices followed by a loop of gradient ascent on the allocations and savings (\Cref{alg:nested_gda_on_qfunc}).
The requisite gradients were computing via auto-differentiation using JAX \cite{jax2018github}, for which we observed better numerical stability than analytically derived gradients, as can often be the case \cite{Griewank2012NumericalStability}.

In both experiments, to check whether the optimal value function was found, we measured the exploitability of the market, meaning the distance between the recCE computed and the actual recCE.
To do so, we checked two conditions:
1) whether each buyer's expected utility was maximized at the computed allocation and savings, at the prices outputted by the algorithm, and 2) whether the market always cleared.
In both settings, we extracted the greedy policy from the value function computed by value iteration, and unrolled it across time to obtain the greedy actions
$({\allocation}^{(\iter)}, {\saving}^{(\iter)},  {\price}^{(\iter)})$ at each state $\state[\iter]$.
We then computed the cumulative utility of these allocation and savings, i.e., for all $\buyer \in \buyers$, $\estutil[\buyer] \doteq \sum_{\iter=0}^{\numiters} \discount^\iter\util[\buyer](\allocation[\buyer]^{(\iter)})$.
We compared these values to the expected maximum utility $\util[i]^*$,
given the prices and the other buyers' allocations 
computed by our algorithm.
We report the normalized distance between these two values, $\estutil[\buyer]$ and $\util[i]^*$, which we call the \mydef{normalized distance to utility maximization (UM)}.
For example, in the case of two buyers, the normalized distance to UM $=\frac{||(\estutil[1],\estutil[2])-(\util[1]^\star,\util[2]^\star)||_2}{||(\util[1]^*,\util[2]^*)||_2}$.
Finally, we also measured excess demand, which we took as the \mydef{distance to market clearance (MC)}, i.e.,  $\frac{1}{\numiters}\sum_{\iter=1}^{ \numiters} ||\sum_{\buyer \in \buyers}\allocation[\buyer]^{(\iter)} - \supply[][\iter]||_2$.

In the experiment with smaller markets and without interest rates, \Cref{fig:avg_values_small} depicts the average value of the value function across a sample of states as it varies with time, and \Cref{table:exploitability_small} records the exploitability of the recCE found by \Cref{alg:value_iter_fisher}with nested GDA.
For all three class of utility functions, not only do the value functions converge, exploitability is also sufficiently minimized, as all the buyer utilities are maximized and the market always clears. 

\begin{figure}[t]
\begin{subfigure}[h]{0.4\linewidth}
\vspace{0cm}
\centering
\includegraphics[width=50mm]{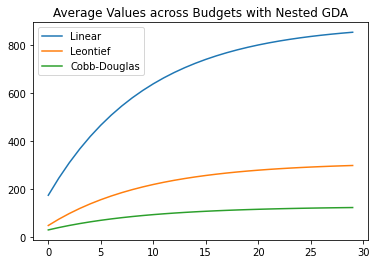}
\captionof{figure}{Average values across budgets.}
\label{fig:avg_values_small}
\end{subfigure}\hfill
\begin{subfigure}[h]{0.55\linewidth}
\vspace{0cm}
\centering
\begin{tabular}{|c|c|c|} 
\hline
\begin{tabular}[c]{@{}c@{}}Utility \\Class\end{tabular} & \begin{tabular}[c]{@{}c@{}}Distance to UM \end{tabular} & \begin{tabular}[c]{@{}c@{}}Distance to MC \end{tabular}  \\ 
\hline
Linear & $0.011$   & $0.010$                                                             \\ 
\hline
Leontief &$0.056$  & $0.010$                                                                \\ 
\hline
Cobb-Douglas& $0.006$   & $0.010$                                                                \\
\hline
\end{tabular}
    \caption{Exploitability of recCE found by \Cref{alg:value_iter_fisher} with nested GDA in stochastic Fisher markets with savings but without interest rates.}
    \label{table:exploitability_small}
\end{subfigure}
\caption{Stochastic Fisher markets with savings but without interest rates.}
\end{figure}


In the experiment with larger markets and probabilistic interest rates (\Cref{fig:avg_values_big}, \Cref{table:exploitability_big}), in linear and Cobb-Douglas markets, the value functions converge, and exploitability is sufficiently minimized.
In Leontief markets, however, although the value function converges and the markets almost clear, the buyers' utilities are not fully maximized; instead, the cumulative utilities obtained are less than half of their expected maximum utilities.
The difficulty in this case likely arises from the fact that the Leontief utility function is not differentiable, so the min-max optimization problem for Leontief markets is neither smooth nor convex-concave, making it difficult, if not impossible, for nested GDA to find even a stationary point,
since gradient ascent is not guaranteed to converge to a stationary point of a function that is non-convex-non-smooth \cite{jordan2022complexity}.


\begin{figure}[t]
\begin{subfigure}[h]{0.4\linewidth}
\vspace{0cm}
\centering
\includegraphics[width=50mm]{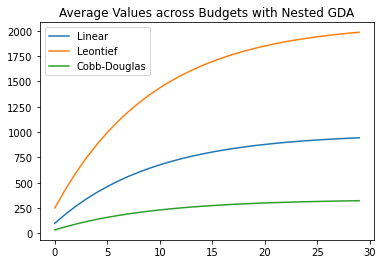}
\captionof{figure}{Average values across budgets.}
\label{fig:avg_values_big}
\end{subfigure}\hfill
\begin{subfigure}[h]{0.55\linewidth}
\vspace{0cm}
\centering
\begin{tabular}{|c|c|c|} 
\hline
\begin{tabular}[c]{@{}c@{}}Utility \\Class\end{tabular} & \begin{tabular}[c]{@{}c@{}}Distance to UM \end{tabular} & \begin{tabular}[c]{@{}c@{}}Distance to MC \end{tabular}  \\ 
\hline
Linear & $0.040$   & $0.009$                                                             \\ 
\hline
Leontief &$0.463$  & $0.009$                                                                \\ 
\hline
Cobb-Douglas& $0.017$   & $0.009$     \\
\hline
\end{tabular}
    \caption{Exploitability of recCE found by \Cref{alg:value_iter_fisher} with nested GDA in stochastic Fisher markets with savings and probabilistic interest rates.}
    \label{table:exploitability_big}
\end{subfigure}
\caption{Stochastic Fisher markets with savings and probabilistic interest rates.}
\end{figure}

%% file: conclusion.tex
\section{Conclusion}

In this paper, we proved the existence of recursive Stackelberg equilibria in zero-sum stochastic Stackelberg games, provided necessary and sufficient conditions for a policy profile to be a recursive Stackelberg equilibrium, and showed that a Stackelberg equilibrium can be computed in (weakly) polynomial time via value iteration.
We then showed how these results can be applied to solve for the recursive competitive equilibria of a stochastic Fisher market with savings, by constructing a zero-sum stochastic Stackelberg game whose recursive Stackelberg equilibria are recursive competitive equilibria of the market.
Finally, we applied this methodology to solve several markets experimentally.

Future work in this space could try using deep reinforcement learning methods to learn better (i.e., nonlinear) representations of the value functions.
It is also conceivable that deep reinforcement learning methods would be able to learn policies that scale up to larger markets than those considered here, as convergence of neural networks to globally optimal solutions has been observed empirically, even in highly non-convex problems \cite{li2018visualizing}, suggesting that they could potentially overcome the difficulties that first-order methods face in the presence of non-smoothness and non-concavity. 
Finally, our experimental results, in which our methods succeed at solving zero-sum \emph{convex-non-concave\/} stochastic Stackelberg games, suggest the need for and possibility of theoretical advances solving convex-non-concave min-max optimization problems, as little is understood beyond non-convex-concave domains.




%% file: ack.tex
\begin{ack}
This work was supported by NSF Grant CMMI-1761546.
\end{ack}

%% file: checklist.tex

\section*{Checklist}


\begin{enumerate}

\item For all authors...
\begin{enumerate}
  \item Do the main claims made in the abstract and introduction accurately reflect the paper's contributions and scope?
    \answerYes{}
  \item Did you describe the limitations of your work?
    \answerYes{Yes, details are discussed in conclusion, and experiments sections.}
  \item Did you discuss any potential negative societal impacts of your work?
    \answerNA{}
  \item Have you read the ethics review guidelines and ensured that your paper conforms to them?
    \answerYes{}
\end{enumerate}

\item If you are including theoretical results...
\begin{enumerate}
  \item Did you state the full set of assumptions of all theoretical results?
    \answerYes{Assumptions are summarized in \Cref{assum:main}}
        \item Did you include complete proofs of all theoretical results? 
    \answerYes{Yes, all proofs can be found in the appendix.}
\end{enumerate}

\item If you ran experiments...
\begin{enumerate}
  \item Did you include the code, data, and instructions needed to reproduce the main experimental results (either in the supplemental material or as a URL)?
    \answerYes{Yes, code repo can be found here: \coderepo.}
  \item Did you specify all the training details (e.g., data splits, hyperparameters, how they were chosen)?
    \answerYes{Details can be found in section 5 and Appendix E.}
        \item Did you report error bars (e.g., with respect to the random seed after running experiments multiple times)?
    \answerYes{For experiments run multiple times, i.e., zero-sum convex constraints, exploitability always reached zero, and there was hence no error to report in all randomly initialized examples.}
        \item Did you include the total amount of compute and the type of resources used (e.g., type of GPUs, internal cluster, or cloud provider)?
    \answerYes{Details can be found in Appendix E.}
\end{enumerate}

\item If you are using existing assets (e.g., code, data, models) or curating/releasing new assets...
\begin{enumerate}
  \item If your work uses existing assets, did you cite the creators?
    \answerYes{Details can be found in Appendix E.}
  \item Did you mention the license of the assets?
    \answerYes{Details can be found in Appendix E.}
  \item Did you include any new assets either in the supplemental material or as a URL?
    \answerNA{}
  \item Did you discuss whether and how consent was obtained from people whose data you're using/curating?
    \answerNA{}
  \item Did you discuss whether the data you are using/curating contains personally identifiable information or offensive content?
    \answerNA{}
\end{enumerate}

\item If you used crowdsourcing or conducted research with human subjects...
\begin{enumerate}
  \item Did you include the full text of instructions given to participants and screenshots, if applicable?
    \answerNA{}
  \item Did you describe any potential participant risks, with links to Institutional Review Board (IRB) approvals, if applicable?
    \answerNA{}
  \item Did you include the estimated hourly wage paid to participants and the total amount spent on participant compensation?
    \answerNA{}
\end{enumerate}

\end{enumerate}

%% file: appendix/algo.tex
\section{Pseudo-Code for Algorithms}\label{sec_app:algorithms}

\begin{algorithm}[H]
\caption{Value Iteration (with Min-Max Oracle)}
\textbf{Inputs:} $\states, \outerset, \innerset, \reward, \constr, \trans, \discount, \numiters$\\
\textbf{Outputs:} $\statevalue[{\numiters}]$
\label{alg:value_iter}
\begin{algorithmic}[1]
\State Initialize $\statevalue[0]$ arbitrarily, e.g. $\statevalue[0]=0$
\For{$\iter = 1, \hdots, \numiters$}
    \For{$\state \in \states$}
        \State $\displaystyle
        \statevalue[\iter+1] (\state) = \min_{\outer \in \outerset} 
        \max_{\inner \in \innerset: \constr(\state, \outer, \inner) \geq \zeros} \mathop{\E}_{\staterv^\prime \sim \trans(\cdot \mid \state, \outer, \inner)} \left[ \reward(\state, \outer, \inner) +  \discount \statevalue[\iter](\staterv^\prime)  \right]$ \label{algeq:value_iter}
    \EndFor
\EndFor
\State \Return $\statevalue[{\numiters}]$ 
\end{algorithmic}
\end{algorithm}


\begin{algorithm}[H]
\caption{Nested GDA for stochastic Fisher Markets with saving}
\label{alg:nested_gda_on_qfunc}
\textbf{Inputs:}  ${\statevalue}, \type, \budget, \supply,   \learnrate[\price], \learnrate[\allocation], \numiters_{\price}, \numiters_{\allocation}, \price^{(0)}, \allocation^{(0)}, \saving^{(0)}$\\ 
\textbf{Output:} $(\price^{(\iter)}, \allocation^{(\iter)})_{\iter = 1}^\numiters$ 
\begin{algorithmic}[1]
\For{$\iter = 1, \hdots, \numiters_{\price}$}
    \For{$s=1, \hdots, \numiters_{\allocation}$}
    \State For all $\buyer \in \buyers$,  $\displaystyle \allocation[\buyer]^{(\iter)} =  
    \allocation[\buyer]^{(\iter)} + \learnrate[\allocation]\left(
    \frac{\budget[\buyer]-\saving[\buyer]^{(\iter)}}{\util[\buyer]\left(\allocation[\buyer]^{(\iter)}; \type[\buyer]\right)} \grad[{\allocation[\buyer]}] \util[\buyer]\left(\allocation[\buyer]^{(\iter)};\type[\buyer] \right) \right) $ 
    \State For all $\buyer \in \buyers$, $\displaystyle \saving[\buyer]^{(\iter)} =  
    \saving[\buyer]^{(\iter)}+ \learnrate[\allocation] \left(
    -\log (\util[\buyer]({\allocation[\buyer]}^{(\iter)}; \type[\buyer])) + \discount \frac{\partial \statevalue (\type, \budget, \supply)}{\partial \budget[\buyer]}
    \right)$
    \State $
    (\allocation[][][\iter], \saving[][\iter])  = \project[\{({\allocation},{\saving)} \in \R_+^{\numbuyers \times \numgoods} \times \R^\numbuyers_+: {\allocation}\cdot \price^{(\iter-1)} + {\saving} \leq {\budget}\}] \left((\allocation[][][\iter], \saving[][\iter]) \right)$
    \EndFor
    \State $\price^{(\iter)} = \project[\R_+^\numgoods]\left(
    \price^{(\iter-1)} - \learnrate[\price](\ones - \sum_{\buyer \in \buyers} \allocation[\buyer][][\iter])
    \right)$
\EndFor
\State \Return $(\price^{(\iter)}, \allocation^{(\iter)})_{\iter = 1}^{\numiters_{\price}}$ 
\end{algorithmic}
\end{algorithm}

%% file: appendix/properties.tex
\section{Omitted Results and Proofs \Cref{sec:properties}}\label{sec_app:properties}

We first note the following fundamental relationship between the state-value and action-value functions which is an analog of Bellman’s Theorem \cite{bellman1952theory} and which follows from their definitions:
\begin{restatable}{theorem}{thmbellman}\label{thm:bellman}
Given a stochastic min-max Stackelberg game $(\states, \outerset, \innerset, \initstates, \reward, \constr, \trans, \discount)$, for all $\statevalue \in \statevalfuncs$, $\actionvalue \in \actionvalfuncs$, $\policy[\outer] \in \outerset^\states$, and $\policy[\inner] \in \innerset^\states$, $\statevalue = \statevalue[][{\policy[\outer]}][{\policy[\inner]}]$ and $\actionvalue = \actionvalue[][{\policy[\outer]}][{\policy[\inner]}]$ iff:
\begin{align}
    &\statevalue[](\state) = \actionvalue[](\state, \policy[\outer](\state), \policy[\inner](\state))\label{eq:belman_q_to_v}\\
    &\actionvalue[](\state, \outer, \inner) = \mathop{\E}_{\staterv^\prime \sim \trans(\cdot \mid \state, \outer, \inner)} \left[ \reward(\state, \outer, \inner) +  \discount \statevalue[](\staterv^\prime)  \right]\label{eq:belman_v_to_q}
\end{align}
\end{restatable}

\begin{proof}[Proof of \Cref{thm:bellman}]
By the definition of the state value function we have $\statevalue[][{\policy[\outer]}][{\policy[\inner]}][\player] = \actionvalue[][][][\player](\state, \policy[\outer](\state), \policy[\inner](\state))$, hence by \Cref{eq:belman_q_to_v} we must have that $\statevalue[][][][\player] = \statevalue[][{\policy[\outer]}][{\policy[\inner]}][\player]$. Additionally, by \Cref{eq:belman_v_to_q} and the definition of the action-value functions this also implies that $\actionvalue[][][][\player](\state, \outer, \inner) = \actionvalue[][{\policy[\outer]}][{\policy[\inner]}][\player](\state, \outer, \inner)$
\end{proof}

\thmfpbelloptisstackelberg*
\begin{proof}[Proof of \Cref{thm:fp_bellopt_is_stackelberg}]

We prove one direction, the other direction follows symmetrically.
(Fixed Point $\implies$ recursive Stackelberg equilibrium)
Suppose that a value function $\statevalue[][{\policy[\outer]^*}][{\policy[\inner]^*}]$ which is induced by a policy profile $(\policy[\outer]^*, \policy[\inner]^*)$ is a fixed point of $\bellopt$, we then have for all states $\state \in \states$:
\begin{align}
    \statevalue[][{\policy[\outer]^*}][{\policy[\inner]^*}] &= \left(\bellopt\statevalue[][{\policy[\outer]^*}][{\policy[\inner]^*}]\right) (\state) \\
    &= \min_{\outer \in \outerset} \max_{\inner \in \innerset : \constr(\state, \outer, \inner)\geq \zeros} \mathop{\E}_{\staterv^\prime \sim \trans(\cdot \mid \state, \outer, \inner)} \left[ \reward(\state, \outer, \inner) +  \discount \statevalue[][{\policy[\outer]^*}][{\policy[\inner]^*}](\staterv^\prime)  \right]\\
    &= \min_{\outer \in \outerset} \max_{\inner \in \innerset : \constr(\state, \outer, \inner)\geq \zeros} \actionvalue[][{\policy[\outer]^*}][{\policy[\inner]^*}](\state, \outer, \inner)
\end{align}

\noindent
Hence, by \Cref{lemma:stackelberg_action_value}, $(\policy[\outer]^*, \policy[\inner]^*)$ is recursive Stackelberg equilibrium.
\end{proof}

\lemmaminmaxlipschitz*

\begin{proof}[Proof of \Cref{lemma:minmax_lipschitz}]
Let $(\outer^*, \inner^*)$ be a Stackelberg equilibrium of  $\min_{\outer \in \outerset} \max_{\inner \in \innerset: \constr(\outer, \inner) \geq \zeros} f(\outer, \inner)$, and $(\outer^\prime, \inner^\prime)$ be a Stackelberg equilibrium of $\min_{\outer \in \outerset} \max_{\inner \in \innerset : \constr(\outer, \inner)
\geq \zeros} h(\outer, \inner)$. Additionally, let $\bar{y} \in \argmax_{\inner \in \innerset: \constr(\outer^\prime, \inner) \geq 0} f(\outer^\prime, \inner)$, then by the the definition of a Stackelberg equilibrium, we have $f(\outer^*, \inner^*) = \min_{\inner \in \innerset} \max_{\outer \in \outerset: \constr(\outer, \inner) \geq \zeros} f(\outer,\inner) \leq \max_{\inner \in \innerset: \constr(\outer^\prime, \inner) \geq 0} f(\outer^\prime, \inner) = f(\outer^\prime, \bar{\inner})$, and $h(\outer^\prime, \inner^\prime) = \max_{\inner \in \innerset : \constr(\outer^\prime, \inner) \geq 0} h(\outer^\prime, \inner)\geq h(\outer^\prime, \bar{\inner})$.

\noindent 
Suppose that $\min_{\outer \in \outerset} \max_{\inner \in \innerset : \constr(\outer, \inner)\geq \zeros} f(\outer, \inner) \geq \min_{\outer \in \outerset} \max_{\inner \in \innerset : \constr(\outer, \inner) \geq \zeros}, h(\outer, \inner)$
this gives us:
\begin{align}
&\left|\min_{\outer \in \outerset} \max_{\inner \in \innerset : \constr(\outer, \inner)\geq \zeros} f(\outer, \inner) - \min_{\outer \in \outerset} \max_{\inner \in \innerset : \constr(\outer, \inner)
\geq \zeros} h(\outer, \inner)\right|  \\
&= \left| \obj(\outer^*, \inner^*) - h(\outer^\prime, \inner^\prime)\right|\\
&\leq  \left| \obj(\outer^\prime, \bar{\inner}) - h(\outer^\prime, \inner^\prime)\right|\\
&\leq  \left| \obj(\outer^\prime, \bar{\inner}) - h(\outer^\prime, \bar{\inner})\right|\\
&\leq \max_{(\outer, \inner) \in \outerset \times \innerset }\left| f(\outer, \inner) -  h(\outer, \inner)\right|
\end{align}
\noindent
The opposite case follows similarly by symmetry.
\end{proof}

\thmcontractionmapping*

\begin{proof}[Proof of \Cref{thm:contraction_mapping}]
We will show that $\bellopt$ is a contraction mapping, which then by Banach fixed point theorem establish the result.
Let $\statevalue[], {\statevalue[]}^\prime \in \statevalfuncs$ be any two state value functions and $\actionvalue, {\actionvalue}^\prime \in \actionvalfuncs$ be the respective associated action-value functions. We then have:
\begin{align}
    &\left\|\bellopt \statevalue[] - \bellopt {\statevalue[]}^\prime\right\|_{\infty} \\ 
    &\leq \max_{\state \in \states} \left|\min_{\outer \in \outerset} \max_{\inner \in \innerset : \constr(\state, \outer, \inner) \geq  \zeros} \actionvalue(\state, \outer, \inner) - \min_{\outer \in \outerset} \max_{\inner \in \innerset : \constr(\state, \outer, \inner) \geq  \zeros} {\actionvalue}^{\prime}(\state, \outer, \inner) \right|\\
    &\leq \max_{\state \in \states} \max_{(\outer, \inner) \in \outerset \times \innerset} \left| \actionvalue(\state, \outer, \inner) - {\actionvalue}^{\prime}(\state, \outer, \inner) \right| && \text{(\Cref{lemma:minmax_lipschitz})}\\
    &\leq \max_{\state \in \states} \max_{(\outer, \inner) \in \outerset \times \innerset} \left|\mathop{\E}_{\staterv^\prime \sim \trans(\cdot \mid \state, \outer, \inner)} \left[ \reward(\state, \outer, \inner) +  \discount \statevalue[](\staterv^\prime)  \right] - \mathop{\E}_{\staterv^\prime \sim \trans(\cdot \mid \state, \outer, \inner)} \left[ \reward(\state, \outer, \inner) +  \discount {\statevalue[]}^\prime(\staterv^\prime)  \right] \right|\\
    &\leq \max_{\state \in \states} \max_{(\outer, \inner) \in \outerset \times \innerset} \left|\mathop{\E}_{\staterv^\prime \sim \trans(\cdot \mid \state, \outer, \inner)} \left[\discount \statevalue[](\staterv^\prime)  - \discount {\statevalue[]}^\prime(\staterv^\prime)  \right] \right|\\
    &\leq \discount \max_{\state \in \states} \max_{(\outer, \inner) \in \outerset \times \innerset} \left|\mathop{\E}_{\staterv^\prime \sim \trans(\cdot \mid \state, \outer, \inner)} \left[ \statevalue[](\staterv^\prime)  -  {\statevalue[]}^\prime(\staterv^\prime)  \right] \right|\\
    &\leq \discount \max_{\state \in \states} \max_{(\outer, \inner) \in \outerset \times \innerset} \left| \statevalue[](\state)  -  {\statevalue[]}^\prime(\state)  \right|\\
    &= \discount \left\|\statevalue[] - {\statevalue[]}^\prime \right\|_\infty
\end{align}

Since $\discount \in (0,1)$, $\bellopt$ is a contraction mapping.

\end{proof}

\thmexistencestackelberg*

\begin{proof}[Proof of \Cref{thm:existence_stackelberg}] 
    By combining  \Cref{thm:contraction_mapping} and the Banach fixed point theorem \cite{banach1922operations}, we obtain that a fixed point of $\bellopt$ exists. Hence, by \Cref{thm:fp_bellopt_is_stackelberg}, a recursive Stackelberg equilibrium of $(\states, \outerset, \innerset, \initstates, \reward, \constr, \trans, \discount)$ exists and the value function induced by all recursive Stackelberg equilibria is the same, i.e., the optimal value function is unique. Additionally, by the second part of the Banach fixed point theorem, we must then also have $\lim_{\iter \to \infty} \statevalue[\iter] = \statevalue[][{\policy[\outer]^*}][{\policy[\inner]^*}]$.
\end{proof}

For any $\actionvalue \in \actionvalfuncs$, we define a \mydef{greedy policy profile}  with respect to $\actionvalue$ as a policy profile $(\policyq[\outer], \policyq[\inner])$ such that $ \policyq[\outer] \in \argmin_{\outer\in \outerset} \max_{\inner \in \innerset : \constr(\state, \outer, \inner) \geq \zeros} \actionvalue(\state, \outer,\inner)$ and $\policyq[\inner] \in   \argmax_{\inner \in \innerset : \constr(\state, \policyq[\outer](\outer), \inner) \geq \zeros} \actionvalue(\state, \policyq[\outer](\outer),\inner)$. The following lemma provides a progress bound for each iteration of value iteration which is expressed in terms of the value function associated with the greedy policy profile.

\thmvalueiter*
\begin{proof}[Proof of \Cref{thm:value_iter}]
First note that by \Cref{assum:main}, we have that $\left\|\statevalue[][{\policy[\outer]^*}][{\policy[\inner]^*}]\right\|_\infty \leq \frac{\rewardbound}{1-\discount}$.
Applying the operator $\bellopt$ repeatedly and using the fact that $\statevalue[][{\policy[\outer]^*}][{\policy[\inner]^*}]=\bellopt{\statevalue[][{\policy[\outer]^*}][{\policy[\inner]^*}]}$ from \Cref{thm:fp_bellopt_is_stackelberg}, we obtain
\begin{align}
    &\|\statevalue[k]-\statevalue[][{\policy[\outer]^*}][{\policy[\inner]^*}]\|_\infty\\
    &=\|(\bellopt)^{k}\statevalue[0]-(\bellopt)^{k}\statevalue[][{\policy[\outer]^*}][{\policy[\inner]^*}]\|_\infty\\
    &\leq \discount^{k}\|\statevalue[0]-\statevalue[][{\policy[\outer]^*}][{\policy[\inner]^*}]\|_\infty\\
    &=\discount^k \|\statevalue[][{\policy[\outer]^*}][{\policy[\inner]^*}]\|_\infty 
   \label{eq:value_0}\\
    &\leq \discount^k \frac{\rewardbound}{1-\discount}
\end{align}
where \Cref{eq:value_0} was obtained as $\statevalue[0]=0$
Since $1-x\leq e^{-x}$ for any $x \in \R$, we have
\begin{align*}
    \discount^k = (1-(1-\discount))^k
    \leq (e^{-(1-\discount)})^k 
    \leq e^{-(1-\discount)k}
\end{align*}
Thus, for any $\state\in \states$
\begin{align*}
    \statevalue[k](\state)  - \statevalue[][{\policy[\outer]^*}][{\policy[\inner]^*}](\state)
    &\leq \|\statevalue[k]-\statevalue[][{\policy[\outer]^*}][{\policy[\inner]^*}]\|_\infty\\
    &= \discount^k \frac{\rewardbound}{1-\discount}\\
    &\leq e^{-(1-\discount)k}\frac{\rewardbound}{1-\discount}
\end{align*}

Thus it suffices to solve for $k$ such that
\begin{align*}
    e^{-(1-\discount)k}\frac{\rewardbound}{1-\discount} \leq \varepsilon \enspace .
\end{align*}
which concludes the proof.
\end{proof}

%% file: appendix/benveniste.tex
\section{Omitted Results and Proofs \Cref{sec:benveniste}}\label{sec_app:benveniste}

Our characterization of the subdifferential of the value function associated with a Stackelberg equilibrium w.r.t.\ its parameters relies on a slightly generalized version of the subdifferential envelope theorem (\Cref{thm:envelope_sd}, \Cref{sec_app:benveniste}) of \citeauthor{goktas2021minmax} \cite{goktas2021minmax}, which characterizes the set of subdifferentials of parametrized constrained optimization problems, i.e., the set of subgradients w.r.t.\ $\outer$ of $\obj^*(\outer) = \max_{\inner \in \innerset: h(\outer, \inner) \geq 0} \obj(\outer, \inner)$.
In particular, we note that \citeauthor{goktas2021minmax}'s proof goes through even without assuming the concavity of $\obj(\outer, \inner), h_{1}(\outer, \inner), \hdots, h_{\numconstrs}(\outer, \inner)$ in $\inner$, for all $\outer \in \outerset$. 

\begin{theorem}[Subdifferential Envelope Theorem]
\label{thm:envelope_sd}
Consider the function $\obj^*(\outer) = \max_{\inner \in \innerset: \h(\outer, \inner) \geq 0} \obj(\outer, \inner) $ where $\obj: \outerset \times \innerset \to \R$, and $\h:  \outerset \times \innerset \to \R^\numconstrs$.
Let $\innerset^*(\outer) = \argmax_{\inner \in \innerset:  \h(\outer, \inner) \geq \zeros} \obj(\outer, \inner)$. Suppose that 
1.~$\obj(\outer, \inner), h_1(\outer, \inner), \ldots, h_\numconstrs(\outer, \inner)$ are continuous in $(\outer, \inner)$ and $\obj^*$ convex in $\outer$;
2.~$\grad[\outer] f, \grad[\outer] h_1, \ldots, \grad[\outer] h_\numconstrs$ are continuous in $(\outer, \inner)$; 
3.~ $\innerset$ is non-empty and compact, and 4.~(Slater's condition)
$\forall \outer \in \outerset, \exists \widehat{\inner} \in \innerset$ s.t.\ $g_{\numconstr}(\outer, \widehat{\inner}) > 0$, for all $\numconstr = 1, \ldots, \numconstrs$.
Then, $\obj^*$ is subdifferentiable and at any point $\widehat{\outer} \in \outerset$, $\subdiff[\outer] \obj^*(\widehat{\outer}) =$
\begin{align}
\label{eq:envelope_deriv_sd}
         \mathrm{conv} \left( \bigcup_{\inner^*(\widehat{\outer}) \in \innerset^*(\widehat{\outer})}      \bigcup_{\langmult[\numconstr]^*(\widehat{\outer}, \inner^*(\widehat{\outer})) \in \langmults[\numconstr]^*(\widehat{\outer}, \inner^*(\widehat{\outer}))} \left\{ \grad[\outer] \obj\left( \widehat{\outer}, \inner^{*}(\widehat{\outer})\right) + \sum_{\numconstr = 1}^{\numconstrs} \langmult[\numconstr]^*(\widehat{\outer}, \inner^*(\widehat{\outer})) \grad[\outer] \constr[k]\left(\widehat{\outer}, \inner^{*}(\widehat{\outer})\right) \right\}
    \right) \enspace ,
\end{align}
\noindent
where $\mathrm{conv}$ is the convex hull operator and $\langmult^*(\widehat{\outer}, \inner^*(\widehat{\outer})) = \left(\langmult[1]^*(\widehat{\outer}, \inner^*(\widehat{\outer})), \ldots, \langmult[\numconstrs]^*(\widehat{\outer}, \inner^*(\widehat{\outer})) \right)^T \in \langmults^*(\widehat{\outer}, \inner^*(\widehat{\outer}))$ are the optimal KKT multipliers associated with $\inner^{*}(\widehat{\outer}) \in \innerset^*(\widehat{\outer})$.
\end{theorem}

\begin{restatable}[Subdifferential Benveniste-Scheinkman Theorem]{theorem}{thmsubdiffbellman}\label{thm:subdiff_bellman}
Consider the Bellman equation associated with a recursive stochastic optimization problem where $\reward: \states \times \outerset \times \innerset \to \R$, with state space $\states$ and parameter set $\outerset$, and $\discount \in (0,1)$:
\begin{align}
    \statevalue(\state, \outer) = \max_{\inner \in  \innerset: \constr(\state, \outer, \inner) \geq \zeros} \left\{\reward(\state, \outer, \inner) + \discount \mathop{\E}_{\staterv^\prime \sim \trans(\cdot \mid \state, \outer, \inner)} \left[\statevalue(\staterv^\prime, \outer) \right] \right\}
\end{align}
Suppose that \Cref{assum:main} holds, and that
1.~ for all $\state \in \states, \inner \in \innerset$, $\reward(\state, \outer, \inner)$, $\constr[1](\state, \outer, \inner), \hdots, \constr[\numconstrs](\state, \outer, \inner)$ are 
concave in $\outer$, 
2.~ $\grad[\outer] \reward(\state, \outer, \inner), \grad[\outer] \constr[1](\state, \outer, \inner), \ldots, \grad[\outer] \constr[\numconstrs](\state, \outer, \inner), \grad[\outer] \trans(\state^\prime \mid \state, \outer, \inner)$ are continuous in $(\state, \state^\prime, \outer, \inner)$,
3.~ $\statevalue(\state, \outer)$ is convex in $\outer$
4.~ Slater's condition holds for the optimization problem, i.e., $\forall \outer \in \outerset, \state \in \states, \exists \widehat{\inner} \in \innerset$ s.t.\ $g_{\numconstr}(\state, \outer, \widehat{\inner}) > 0$, for all $\numconstr = 1, \ldots, \numconstrs$. \\
Let $\innerset^*(\state, \outer) = \max_{\inner \in  \innerset: \constr(\state, \outer, \inner) \geq \zeros} \left\{\reward(\state, \outer, \inner) + \discount \mathop{\E}_{\state^\prime \sim \trans(\cdot \mid \state, \outer, \inner)} \left[\statevalue(\state^\prime, \outer) \right] \right\}$, then $\statevalue$ is subdifferentiable and $\subdiff[\outer]\statevalue(\state, \hat{\outer}) =$
\begin{align}
    \mathrm{conv} \left( \bigcup_{\inner^*(\state, \widehat{\outer}) \in \innerset^*(\state, \widehat{\outer})}      \bigcup_{\langmult[\numconstr]^*(\state, \widehat{\outer}, \inner^*(\state, \widehat{\outer})) \in \langmults[\numconstr]^*(\state, \widehat{\outer}, \inner^*(\widehat{\outer}))} \left\{ \grad[\outer] \reward\left(\state,  \widehat{\outer}, \inner^{*}(\state, \widehat{\outer})\right) +  \discount \grad[\outer] \mathop{\E}_{\staterv^\prime \sim \trans(\cdot \mid \state, \hat{\outer}, \inner^{*}(\state, \widehat{\outer}))} \left[  \statevalue(\staterv^\prime, \hat{\outer})\right] \right. \right. \notag \\ 
     \left. \left. 
    + \sum_{\numconstr = 1}^{\numconstrs} \langmult[\numconstr]^*(\state, \widehat{\outer}, \inner^*(\state, \widehat{\outer})) \grad[\outer] \constr[k]\left(\state, \widehat{\outer}, \inner^{*}(\state, \widehat{\outer})\right) \right\}
    \right) \enspace .
\end{align}
Suppose additionally, that for all $\state, \state^\prime \in \states, \outer \in \outerset$, $\inner^*(\state, \outer) \in \innerset^*(\state, \outer)$ $\grad[\outer] \trans(\state^\prime \mid \state, \outer, \inner^*(\state, \outer)) > 0$, then  $\subdiff[\outer]\statevalue(\state, \hat{\outer}) =$
\begin{align}
    &\mathrm{conv} \left( \bigcup_{\inner^*(\state, \widehat{\outer}) \in \innerset^*(\state, \widehat{\outer})}      \bigcup_{\langmult[\numconstr]^*(\state, \widehat{\outer}, \inner^*(\state, \widehat{\outer})) \in \langmults[\numconstr]^*(\state, \widehat{\outer}, \inner^*(\widehat{\outer}))} \left\{ \grad[\outer] \reward\left(\state,  \widehat{\outer}, \inner^{*}(\state, \widehat{\outer})\right) +  \discount \mathop{\E}_{\staterv^\prime \sim \trans(\cdot \mid \state, \hat{\outer}, \inner^{*}(\state, \widehat{\outer}))} \left[ \grad[\outer] \statevalue(\staterv^\prime, \hat{\outer})\right] \right. \right. \notag \\ 
    & \left. \left. + \discount  \mathop{\E}_{\staterv^\prime \sim \trans(\cdot \mid \state, \hat{\outer}, \inner^{*}(\state, \widehat{\outer}))} \left[\statevalue(\staterv^\prime, \hat{\outer}) \grad[\outer] \log\left(\trans(\staterv^\prime \mid \state, \hat{\outer}, \inner^{*}(\state, \widehat{\outer})\right)\right] 
    + \sum_{\numconstr = 1}^{\numconstrs} \langmult[\numconstr]^*(\state, \widehat{\outer}, \inner^*(\state, \widehat{\outer})) \grad[\outer] \constr[k]\left(\state, \widehat{\outer}, \inner^{*}(\state, \widehat{\outer})\right) \right\}
    \right) \enspace .
\end{align}
where $\mathrm{conv}$ is the convex hull operator and $\langmult^*(\state, \widehat{\outer}, \inner^*(\state, ,\widehat{\outer})) = \left(\langmult[1]^*(\state, \widehat{\outer}, \inner^*(\state, \widehat{\outer})), \ldots, \langmult[\numconstrs]^*(\state,\widehat{\outer}, \inner^*(\state, \widehat{\outer})) \right)^T \in \langmults^*(\state,\widehat{\outer}, \inner^*(\state,\widehat{\outer}))$ are the optimal KKT multipliers associated with $\inner^{*}(\state,\widehat{\outer}) \in \innerset^*(\state,\widehat{\outer})$.
\end{restatable}

\begin{proof}[Proof of \Cref{thm:subdiff_bellman}]
From \Cref{thm:envelope_sd}, we obtain the first part of the theorem:
\begin{align}
    &\subdiff[\outer]\statevalue(\state, \hat{\outer}) \\
    &=
    \mathrm{conv} \left( \bigcup_{\inner^*(\state, \widehat{\outer}) \in \innerset^*(\state, \widehat{\outer})}      \bigcup_{\langmult[\numconstr]^*(\state, \widehat{\outer}, \inner^*(\state, \widehat{\outer})) \in \langmults[\numconstr]^*(\state, \widehat{\outer}, \inner^*(\state, \widehat{\outer}))} \left\{ \grad[\outer] \reward\left(\state,  \widehat{\outer}, \inner^{*}(\state, \widehat{\outer})\right) 
    + \discount  \grad[\outer]\mathop{\E}_{\state^\prime \sim \trans(\cdot \mid \state, \outer, \inner^{*}(\state, \widehat{\outer}))} \left[\statevalue(\staterv^\prime, \outer) \right] \right. \right. \notag \\
    & \left. \left.
    + \sum_{\numconstr = 1}^{\numconstrs} \langmult[\numconstr]^*(\state, \widehat{\outer}, \inner^*(\state, \widehat{\outer})) \grad[\outer] \constr[k]\left(\state, \widehat{\outer}, \inner^{*}(\state, \widehat{\outer})\right) \right\}
    \right) \enspace .
\end{align}

By the Leibniz integral rule \cite{Flanders1973LeibnizRule}, the gradient of the expectation can instead be expressed as an expectation of the gradient under continuity of the function whose expectation is taken, in this case $\statevalue$. In particular, if for all $\state, \state^\prime \in \states, \outer \in \outerset$, $\inner^*(\state, \outer) \in \innerset^*(\state, \outer)$ $\grad[\outer] \trans(\state^\prime \mid \state, \outer, \inner^*(\state, \outer)) > 0$ we have: 
 
 \begin{align}
    &\grad[\outer]\mathop{\E}_{\staterv^\prime \sim \trans(\cdot \mid \state, \outer, \inner)} \left[\statevalue(\staterv^\prime, \outer) \right] \\
    & = \grad[\outer] \int_{\z \in \states} \trans(\z \mid \state, \outer, \inner) \statevalue(\z, \outer) d\z\\
    &=  \int_{\z \in \states} \grad[\outer] [\trans(\z \mid \state, \outer, \inner) \statevalue(\z, \outer)] d\z && \text{(Leibniz Integral Rule)}\\
    &= \int_{\z \in \states} \left[ \trans(\z \mid \state, \outer, \inner) \grad[\outer]\statevalue(\z, \outer)  +  \statevalue(\z, \outer) \grad[\outer] \trans(\z \mid \state, \outer, \inner) \right] d\z && \text{(Product Rule)}\\
    &= \int_{\z \in \states} \left[ \trans(\z \mid \state, \outer, \inner) \grad[\outer]\statevalue(\z, \outer)   +  \statevalue(\z, \outer) \trans(\z \mid \state, \outer, \inner) \frac{\grad[\outer] \trans(\z \mid \state, \outer, \inner)}{\trans(\z \mid \state, \outer, \inner)} \right] d\z && \text{($\trans(\z \mid \state, \outer, \inner) > 0$)}\\
    &= \int_{\z \in \states} \left[ \trans(\z \mid \state, \outer, \inner) \grad[\outer]\statevalue(\z, \outer)   +  \statevalue(\z, \outer) \trans(\z \mid \state, \outer, \inner) \grad[\outer] \log{\trans(\z \mid \state, \outer, \inner)} \right] d\z\\
    &= \int_{\z \in \states} \left[ \trans(\z \mid \state, \outer, \inner) \grad[\outer]\statevalue(\z, \outer) \right] d\z\  + \int_{\z \in \states} \left[ \statevalue(\z, \outer) \trans(\z \mid \state, \outer, \inner) \grad[\outer] \log{\trans(\z \mid \state, \outer, \inner)} \right] d\z\\
    &= \mathop{\E}_{\staterv^\prime \sim \trans(\cdot \mid \state, \outer, \inner)}\left[ \grad[\outer] \statevalue(\staterv^\prime, \outer) \right] + \mathop{\E}_{\staterv^\prime \sim \trans(\cdot \mid \state, \outer, \inner)}\left[  \statevalue(\staterv^\prime, \outer) \grad[\outer] \log{\trans(\staterv^\prime \mid \state, \outer, \inner)} \right]
 \end{align} 
 
This gives us $\subdiff[\outer]\statevalue(\state, \hat{\outer}) =$
\begin{align}
    &\mathrm{conv} \left( \bigcup_{\inner^*(\state, \widehat{\outer}) \in \innerset^*(\state, \widehat{\outer})}      \bigcup_{\langmult[\numconstr]^*(\state, \widehat{\outer}, \inner^*(\state, \widehat{\outer})) \in \langmults[\numconstr]^*(\state, \widehat{\outer}, \inner^*(\widehat{\outer}))} \left\{ \grad[\outer] \reward\left(\state,  \widehat{\outer}, \inner^{*}(\state, \widehat{\outer})\right) +  \discount \mathop{\E}_{\staterv^\prime \sim \trans(\cdot \mid \state, \hat{\outer}, \inner^{*}(\state, \widehat{\outer}))} \left[ \grad[\outer] \statevalue(\staterv^\prime, \hat{\outer})\right] \right. \right. \notag \\ 
    & \left. \left. + \discount  \mathop{\E}_{\staterv^\prime \sim \trans(\cdot \mid \state, \hat{\outer}, \inner^{*}(\state, \widehat{\outer}))} \left[\statevalue(\staterv^\prime, \hat{\outer}) \grad[\outer] \log\left(\trans(\staterv^\prime \mid \state, \hat{\outer}, \inner^{*}(\state, \widehat{\outer})\right)\right] 
    + \sum_{\numconstr = 1}^{\numconstrs} \langmult[\numconstr]^*(\state, \widehat{\outer}, \inner^*(\state, \widehat{\outer})) \grad[\outer] \constr[k]\left(\state, \widehat{\outer}, \inner^{*}(\state, \widehat{\outer})\right) \right\}
    \right) \enspace .
\end{align}
\end{proof}

Note that in the special case that the probability transition function is representing a deterministic recursive parametrized optimization problem, $
    \statevalue(\state, \outer) = \max_{\inner \in  \innerset: \constr(\state, \outer, \inner) \geq \zeros} \left\{\reward(\state, \outer, \inner) + \discount  \left[\statevalue(\tau(\state, \outer, \inner) , \outer) \right] \right\}
$
i.e., $\trans(\state^\prime \mid \state, \outer, \inner) \in \{0, 1\}$ for all $\state, \state^\prime \in \states, \outer \in \outerset, \inner \in \innerset$,
and $\tau: \states \times \outerset \times \innerset \to \states$ is such that $\tau(\state, \outer, \inner) = \state^\prime$ iff $\trans(\state^\prime \mid \state, \outer, \inner) = 1$, the CSD convexity assumption reduces to the linearity of the deterministic state transition function $\tau$ (Proposition 1 of \cite{atakan2003valfunc}). In this case, the subdifferential of the Bellman equation reduces to 
\begin{align}
&\subdiff[\outer]\statevalue(\state, \hat{\outer}) \\
&=\mathrm{conv} \left( \bigcup_{\inner^*(\state, \widehat{\outer}) \in \innerset^*(\state, \widehat{\outer})}      \bigcup_{\langmult[\numconstr]^*(\state, \widehat{\outer}, \inner^*(\state, \widehat{\outer})) \in \langmults[\numconstr]^*(\state, \widehat{\outer}, \inner^*(\state, \widehat{\outer}))} \left\{ \grad[\outer] \reward\left(\state,  \widehat{\outer}, \inner^{*}(\state, \widehat{\outer})\right) + \discount  \grad[\outer]\tau(\state, \widehat{\outer}, \inner)\grad[\state]\statevalue(\tau(\state, \widehat{\outer},  \inner), \widehat{\outer})  \right. \right.\\
    &\left. \left. + \grad[\outer] \statevalue(\tau(\state, \hat{\outer}), \widehat{\outer}) + \sum_{\numconstr = 1}^{\numconstrs} \langmult[\numconstr]^*(\state, \widehat{\outer}, \inner^*(\state, \widehat{\outer})) \grad[\outer] \constr[k]\left(\state, \widehat{\outer}, \inner^{*}(\state, \widehat{\outer})\right) \right\}
    \right)    
\end{align}

\thmnecessaryconditions*

\begin{proof}[Proof of \Cref{thm:necessary_conditions}]
By \Cref{thm:fp_bellopt_is_stackelberg} and \Cref{thm:contraction_mapping} we know that $(\policy[\outer]^*, \policy[\inner]^*)$ is a recursive Stackelberg equilibrium iff 
\begin{align}\label{eq:fp_optimality}
    \statevalue[][{\policy[\inner]}][{\policy[\inner]^*}](\state) = \left( \bellopt \statevalue[][{\policy[\inner]}][{\policy[\inner]^*}]\right)(\state) \enspace .
\end{align} Note that for any policy profile $(\policy[\outer]^*, \policy[\inner]^*)$ that satisfies $\statevalue[][{\policy[\inner]}][{\policy[\inner]^*}](\state) = \left( \bellopt \statevalue[][{\policy[\inner]}][{\policy[\inner]^*}]\right)(\state)$ by \Cref{lemma:stackelberg_action_value} we have that $(\policy[\outer]^*(\state), \policy[\inner]^*(\state))$ is a Stackelberg equilibrium of 
$$\min_{\outer \in \outerset} \max_{\inner \in \innerset : \constr(\state, \outer, \inner)\geq \zeros} \left\{\reward(\state, \outer, \inner) +  \discount \mathop{\E}_{\staterv^\prime \sim \trans(\cdot \mid \state, \outer, \inner)} \left[ \statevalue[](\staterv^\prime)  \right] \right\}$$
for all $\state \in \states$.

Fix a state $\state \in \states$, under the assumptions of the theorem, the conditions of \Cref{thm:subdiff_bellman} are satisfied and $\marginal(\state, \outer) = \max_{\inner \in \innerset : \constr(\state, \outer, \inner)\geq \zeros} \left\{\reward(\state, \outer, \inner) +  \discount \mathop{\E}_{\staterv^\prime \sim \trans(\cdot \mid \state, \outer, \inner)} \left[ \statevalue[](\staterv^\prime)  \right] \right\}$ is subdifferentiable in $\outer$. Since $\marginal(\state, \outer)$ is convex in $\outer$, and Slater's condition are satisfied by the assumptions of the theorem, the necesssary \emph{and sufficient} conditions for $\policy[\outer]^*(\state)$ to be an optimal solution to $\min_{\outer \in \outerset} \marginal(\state, \outer)$ are given by the KKT conditions \cite{kuhn1951kkt} for $\min_{\outer \in \outerset} \marginal(\state, \outer)$. Note that we can state the first order KKT conditions explicitly thanks to the subdifferential Benveniste-Scheinkman theorem (\Cref{thm:subdiff_bellman}). That is, $\policy[\outer]^*(\state)$ is an optimal solution to $\min_{\outer \in \outerset} \marginal(\state, \outer)$ if there exists $\bm{\mu}^*(\state) \in \R_+^\outernumconstrs$ such that:

\begin{align}
    &\grad[\outer] \lang[\state, {\policy[\outer]^*(\state)}](\inner^*(\state, \policy[\outer]^*(\state)), \langmult^*(\state, \policy[\outer]^*(\state), \inner^*(\state,  \policy[\outer]^*(\state))) + \sum_{\numconstr = 1}^\outernumconstrs \mu^*_\numconstr(\state) \grad[\outer]\outerconstr[\numconstr](\policy[\outer]^*(\state)) = 0\label{eq:optimality_outer_1}\\
    &\mu^*_\numconstr(\state)\outerconstr[\numconstr](\policy[\outer]^*(\state)) = 0 && \forall \numconstr \in [\outernumconstrs]\\
    &\outerconstr[\numconstr](\policy[\outer]^*(\state)) \leq 0 && \forall \numconstr \in [\outernumconstrs]\label{eq:optimality_outer_3}
\end{align}
where $\inner^*(\state, \policy[\outer]^*(\state)) \in \argmax_{\inner \in \innerset : \constr(\state, \policy[\outer]^*(\state), \inner)\geq \zeros} \left\{\reward(\state, \policy[\outer]^*(\state), \inner) +  \discount \mathop{\E}_{\staterv^\prime \sim \trans(\cdot \mid \state, \policy[\outer]^*(\state), \inner)} \left[ \statevalue[](\staterv^\prime)  \right] \right\}$ and \\ $\langmult^*(\state, \policy[\outer]^*(\state), \inner^*(\state,  \policy[\outer]^*(\state))) = \left(\langmult[1]^*(\state, \policy[\outer]^*(\state), \inner^*(\state, \policy[\outer]^*(\state))), \ldots, \langmult[\numconstrs]^*(\state, \policy[\outer]^*(\state), \inner^*(\state, \policy[\outer]^*(\state))) \right)^T \in \langmults^*(\state, \policy[\outer]^*(\state), \inner^*(\state, \policy[\outer]^*(\state)))$ are the optimal KKT multipliers associated with $\inner^{*}(\state, \policy[\outer]^*(\state)) \in \innerset^*(\state, \policy[\outer]^*(\state))$ which are guaranteed to exist since Slater's condition is satisfied for $\max_{\inner \in \innerset : \constr(\state, \outer, \inner)\geq \zeros} \left\{\reward(\state, \outer, \inner) +  \discount \mathop{\E}_{\staterv^\prime \sim \trans(\cdot \mid \state, \outer, \inner)} \left[ \statevalue[](\staterv^\prime)  \right] \right\}$.

Similarly, fix a state $\state \in \states$ and an action for the outer player $\outer \in \outerset$, since Slater's condition is satisfied for $\max_{\inner \in \innerset : \constr(\state, \outer, \inner)\geq \zeros} \left\{\reward(\state, \outer, \inner) +  \discount \mathop{\E}_{\staterv^\prime \sim \trans(\cdot \mid \state, \outer, \inner)} \left[ \statevalue[](\staterv^\prime)  \right] \right\}$, the necessary conditions for $\policy[\inner]^*(\state)$ to be a Stackelberg equilibrium strategy for the inner player at state $\state$ are given by the KKT conditions for $\max_{\inner \in \innerset : \constr(\state, \outer, \inner)\geq \zeros} \left\{\reward(\state, \outer, \inner) +  \discount \mathop{\E}_{\staterv^\prime \sim \trans(\cdot \mid \state, \outer, \inner)} \left[ \statevalue[](\staterv^\prime)  \right] \right\}$. That is, there exists $\langmult^*(\state, \outer) \in \R_+^\numconstrs$ and $\bm{\nu}^*(\state, \outer) \in \R_+^\innernumconstrs$ such that:
\begin{align}
    &\grad[\inner] \lang[\state, {\outer}](\policy[\inner]^*(\state), \langmult^*(\state, \outer)) + \sum_{\numconstr = 1}^\innernumconstrs \nu^*_\numconstr(\state) \grad[\outer]\innerconstr[\numconstr](\policy[\inner]^*(\state)) = 0 \label{eq:optimality_inner_1}\\
    &\constr[\numconstr](\state, \outer, \policy[\inner]^*(\state)) \geq 0 && \forall \numconstr \in [\numconstrs]\\
    &\langmult[\numconstr]^*(\state, \outer) \constr[\numconstr](\state, \outer, \policy[\inner]^*(\state)) = 0 && \forall \numconstr \in [\numconstrs]\\
    &\nu^*_\numconstr(\state, \outer) \grad[\outer]\innerconstr[\numconstr](\policy[\inner]^*(\state)) = 0 && \forall \numconstr \in [\innernumconstrs]\\
    &\innerconstr[\numconstr](\outer) \geq 0 && \forall \numconstr \in [\innernumconstrs]  \label{eq:optimality_inner_5}
\end{align}

Combining the necessary and \emph{sufficient} conditions in \Crefrange{eq:optimality_outer_1}{eq:optimality_outer_3} with the necessary conditions in \Crefrange{eq:optimality_inner_1}{eq:optimality_inner_5}, we obtain the necessary conditions for $(\policy[\outer]^*, \policy[\inner]^*)$ to be a recursive Stackelberg equilibrium.

\end{proof}

If additionally, the objective of the inner player at each state $\state \in \states$, $\reward(\state, \outer, \inner) + \discount \mathop{\E}_{\staterv^\prime \sim \trans(\cdot \mid \state, \outer, \inner)} \left[\statevalue(\staterv^\prime, \outer) \right]$ is concave in $\inner$, then the above conditions become necessary and sufficient. The proof follows exactly the same, albeit the optimality conditions on the inner player's policy become necessary and sufficient.

\begin{theorem}[Recursive Stackelberg Equilibrium Necessary and Sufficient Optimality Conditions]\label{thm:necessary_sufficient_conditions}
Consider a zero-sum stochastic Stackelberg game $\game$, where $\outerset = \{\outer \in \R^\outerdim \mid \outerconstr[1](\outer)\leq 0, \hdots, \outerconstr[\outernumconstrs](\outer) \leq 0 \}$ and $\innerset = \{\inner \in \R^\innerdim \mid 
\innerconstr[1](\inner)\geq 0, \hdots, \innerconstr[\innernumconstrs](\inner)  \geq 0 \}$ are convex.
Let $\lang[\state, \outer](\inner, \langmult) =  \reward(\state, \outer, \inner) + \discount \mathop{\E}_{\staterv^\prime \sim \trans(\cdot \mid \state, \outer, \inner)} \left[\statevalue(\staterv^\prime, \outer) \right] +  \sum_{\numconstr = 1}^\numconstrs \langmult[\numconstr] \constr[\numconstr](\state, \outer, \inner)$ where $\bellopt \statevalue = \statevalue$.
Suppose that  
\Cref{assum:main} holds, and that 
1.~for all $\state \in \states, \inner \in \innerset$, $\max_{\inner \in \innerset: \constr(\state, \outer, \inner) \geq \zeros} \left\{ \reward(\state, \outer, \inner) + \discount \mathop{\E}_{\staterv^\prime \sim \trans(\cdot \mid \state, \outer, \inner)} \left[\statevalue(\staterv^\prime, \outer) \right] \right\}$ is
convex in $\outer$ and $\reward[1](\state, \outer, \inner), \constr[1](\state, \outer, \inner), \hdots, \constr[\numconstrs](\state, \outer, \inner)$ are concave in $\inner$ for all $\outer \in \outerset$ and $\state \in \states$,  
2.~$\grad[\outer] \reward(\state, \outer, \inner), \grad[\outer] \constr[1](\state, \outer, \inner), \ldots, \grad[\outer] \constr[\numconstrs](\state, \outer, \inner)$, $\grad[\inner] \reward(\state, \outer, \inner), \grad[\inner] \constr[1](\state, \outer, \inner), \ldots, \grad[\inner] \constr[\numconstrs](\state, \outer, \inner)$ exist, for all $\state \in \states, \outer \in \outerset, \inner \in \innerset$,
4.~$\trans(\state^\prime \mid \state, \outer, \inner)$ is continuous
and differentiable in $(\outer, \inner)$ and CSD concave in $\inner$, and
5.~Slater's condition holds, i.e., $\forall \state \in \states, \outer \in \outerset, \exists \widehat{\inner} \in \innerset$ s.t.\ $\constr[\numconstr](\state, \outer, \widehat{\inner}) > 0$, for all $\numconstr = 1, \ldots, \numconstrs$ and $\innerconstr[j](\widehat{\inner})  > 0$, for all $j = 1, \hdots, \innernumconstrs$, and $\exists \outer \in \R^\outerdim$ s.t. $\outerconstr[\numconstr](\outer) < 0$ for all $\numconstr =1 \hdots, \outernumconstrs$.
Then, there exists $\bm{\mu}^*: \states \to \R_+^\outernumconstrs$, $\langmult^* : \states \times \outerset \to \R_+^\numconstrs$, and $\bm{\nu}^* : \states \times \outerset \to \R_+^\innernumconstrs$ s.t.\ a policy profile $(\policy[\outer]^*, \policy[\inner]^*) \in \outerset^\states \times \innerset^\states$ is a \recSE{} of $\game$ only if it satisfies the following conditions, for all $\state \in \states$:
\begin{align}
    &\grad[\outer] \lang[\state, {\policy[\outer]^*(\state)}](\policy[\inner]^*(\state), \langmult^*(\state, \policy[\outer]^*(\state))) + \sum_{\numconstr = 1}^\outernumconstrs \mu^*_\numconstr(\state) \grad[\outer]\outerconstr[\numconstr](\policy[\outer]^*(\state)) = 0\\
    &\grad[\inner] \lang[\state, {\policy[\outer]^*( \state)}](\policy[\inner]^*(\state), \langmult^*(\state, \policy[\outer]^*(\state))) + \sum_{\numconstr = 1}^\innernumconstrs \nu^*_\numconstr(\state, \policy[\outer]^*(\state)) \grad[\outer]\innerconstr[\numconstr](\policy[\inner]^*(\state)) = 0
\end{align}
\begin{align}
    &\mu^*_\numconstr(\state)\outerconstr[\numconstr](\policy[\outer]^*(\state)) = 0 &\outerconstr[\numconstr](\policy[\outer]^*(\state)) \leq 0 && \forall \numconstr \in [\outernumconstrs]\\
    &\constr[\numconstr](\state, \policy[\outer]^*(\state), \policy[\inner]^*(\state)) \geq 0 &\langmult[\numconstr]^*(\state, \policy[\outer]^*(\state)) \constr[\numconstr](\state, \policy[\outer]^*(\state), \policy[\inner]^*(\state)) = 0  && \forall \numconstr \in [\numconstrs]\\
    &\nu^*_\numconstr(\state, \policy[\outer]^*(\state)) \grad[\outer]\innerconstr[\numconstr](\policy[\inner]^*(\state)) = 0  &\innerconstr[\numconstr](\policy[\outer]^*(\state)) \geq 0 && \forall \numconstr \in [\innernumconstrs]
\end{align}
\end{theorem}

%% file: appendix/fisher.tex
\section{Omitted Results and Proofs \Cref{sec:fisher}}\label{sec_app:fisher}

Before, we introduce the stochastic Stackelberg game whose recursive Stackelberg equilibria correspond to recursive competitive equilibria of an associated stochastic Fisher market, we introduce the following technical lemma, which provides the necessary and sufficient conditions for an allocation and saving system of a buyer to be expected utility maximizing.

\begin{restatable}{lemma}{lemmasavingsfoc}\label{lemma:savings_foc}
Consider a stochastic Fisher market $\fishermkt$ such that the transition probability function $\trans$ is continuous in $\saving[\buyer]$. For any price system $\price \in \R_+^{\states \times \numgoods}$, a tuple $(\allocation[\buyer]^*, \saving[\buyer]^*) \in \R_+^{\states \times \numgoods} \times \R_+^{ \numgoods}$ consisting of an allocation system and saving system for a buyer $\buyer \in \buyers$, given by a continuous, and homogeneous utility function $\util[\buyer]: \R^\numgoods_+ \times \typespace \to \R$ representing a locally non-satiated preference, is expected utility maximizing constrained by the saving and spending constrains, i.e., $(\allocation[\buyer]^*, \saving[\buyer]^*)$ is the optimal policy solving the following recursive Bellman equation $\budgetval[\buyer](\state) =
    \max_{(\allocation[\buyer], \saving[\buyer]) \in \R^{\numgoods + 1}_+: \allocation[\buyer] \cdot \price^*(\state) + \saving[\buyer] \leq \budget[\buyer]} \left\{ \util[\buyer]\left(\allocation[\buyer], \type[\buyer]\right) 
    + \discount \mathop{\E}_{(\type^\prime, \budget^\prime, \supply^\prime) \sim \trans(\cdot \mid \state, 
    (\saving[\buyer], \saving_{-\buyer}^*(\state)))} \left[ \budgetval[\buyer](\type^\prime, \budget^\prime + \saving[\buyer], \supply^\prime) \right] \right\}
$, only if we have for all states $\state \in \states$,
    $\allocation[\buyer]^*(\type, \budget, \supply) \cdot \price(\type, \budget, \supply) + \saving[\buyer]^*(\type, \budget, \supply) \leq \budget[\buyer]$,
and,  
\begin{align}
    \allocation[\buyer][\good]^*(\state) > 0 &\implies   \frac{\frac{\partial \util[\buyer]}{\partial \allocation[\buyer][\good] }(\allocation[\buyer]^*(\state); \type[\buyer]) 
    }{\price[\good](\state)} =
    \frac{\util[\buyer](\allocation[\buyer]^*(\state); \type[\buyer]) 
    }{\budget[\buyer] - \saving[\buyer]^*(\state)} && \forall \good \in \goods\\
    \saving[\buyer]^*(\state) > 0 &\implies \frac{\partial \budgetval[\buyer]^*}{\partial \budget[\buyer]}(\state)  =\gamma \frac{\partial}{\partial \saving[\buyer]}\mathop{\E}_{(\type^\prime, \budget^\prime, \supply^\prime) } \left[ \budgetval[\buyer]^*(\type^\prime, \budget^\prime + \saving[\buyer]^*(\state), \supply^\prime) \right] 
\end{align}
If additionally, $\util[\buyer]$ is concave and $\trans$ is CSD concave in $(\allocation[\buyer], \saving[\buyer])$, then the above condition becomes also sufficient.
\end{restatable}

\begin{proof}[Proof of \Cref{lemma:savings_foc}]
Fix a buyer $\buyer \in \buyers$. Throughout we use $\budget + \saving[\buyer]$ as shortcut for $\budget + (\saving[\buyer], \zeros[-\buyer])$. Suppose that $\budgetval[\buyer]^*$ solves the following recursive Bellman equation:
\begin{align}
    \budgetval[\buyer](\state) =
    \max_{(\allocation[\buyer], \saving[\buyer]) \in \R^{\numgoods + 1}_+: \allocation[\buyer] \cdot \price^*(\state) + \saving[\buyer] \leq \budget[\buyer]} \left\{ \util[\buyer]\left(\allocation[\buyer], \type[\buyer]\right) 
    + \discount \mathop{\E}_{(\type^\prime, \budget^\prime, \supply^\prime) \sim \trans(\cdot \mid \state, 
    (\saving[\buyer], \saving_{-\buyer}^*(\state)))} \left[ \budgetval[\buyer](\type^\prime, \budget^\prime + \saving[\buyer], \supply^\prime) \right] \right\}
\end{align}

Define the Lagrangian associated with the consumption-saving problem: 
\begin{align}
    \max_{(\allocation[\buyer], \saving[\buyer]) \in \R^{\numgoods + 1}_+: \allocation[\buyer] \cdot \price^*(\state) + \saving[\buyer] \leq \budget[\buyer]} \left\{ \util[\buyer]\left(\allocation[\buyer], \type[\buyer]\right) 
    + \discount \mathop{\E}_{(\type^\prime, \budget^\prime, \supply^\prime) \sim \trans(\cdot \mid \state, 
    (\saving[\buyer], \saving_{-\buyer}^*(\state)))} \left[ \budgetval[\buyer]^*(\type^\prime, \budget^\prime + \saving[\buyer], \supply^\prime) \right] \right\}
\end{align}
as follows:
\begin{align}
    \lang(\allocation[\buyer], \saving[\buyer], \lambda, \bm \mu; \price) &=\util[\buyer]\left(\allocation[\buyer]; \type[\buyer] \right) 
    + \discount \mathop{\E}_{(\type^\prime, \budget^\prime, \supply^\prime) \sim \trans(\cdot \mid \type, \budget, \supply, 
    (\saving[\buyer], \saving_{-\buyer}^*(\state)))} \left[ \budgetval[\buyer]^*(\type^\prime, \budget^\prime + \saving[\buyer], \supply^\prime) \right]\notag\\
    &+ \lambda (\budget[\buyer] - \allocation[\buyer] \price) + \sum_{\good \in \goods} \mu_\good \allocation[\buyer][\good] + \mu_{\numgoods + 1} \saving[\buyer]
\end{align}
    $$
    \budgetval[\buyer]^*(\type, \budget, \supply) = 
        \max_{(\allocation[\buyer], \saving[\buyer]) \in \R^{\numgoods + 1}_+: \allocation[\buyer] \cdot \price^*(\budget) + \saving[\buyer] \leq \budget[\buyer]}  \util[\buyer](\allocation[\buyer], \type[\buyer]) + \discount  \mathop{\E}_{(\type^\prime, \budget^\prime, \supply^\prime) } \left[ \budgetval[\buyer]^*(\type^\prime, \budget^\prime + \saving[\buyer], \supply^\prime) \right] \enspace .
    $$
    Assume that for any state $\state \in \states$, we have $\budget[\buyer] > 0$. We can ignore states such that $\budget[\buyer] > 0$ since at those states the buyer cannot be allocated goods, and can also not put aside savings.
    Then, Slater's condition holds and the necessary first order optimality conditions for an allocation $\allocation[\buyer]^*(\state) \in \R_+^{\numgoods}$, saving $\saving[\buyer]^*(\state) \in \R_+$ and associated Lagrangian multipliers $\lambda^*(\state) \in \R_+$, $\bm \mu^*(\state) \in \R^{\numgoods + 1}$ to be optimal for any prices $\price(\state) \in \R^\numgoods_+$ and state $\state \in \states$ are given by the following pair of KKT conditions \cite{kuhn1951kkt} for all $\good \in \goods$:
\begin{align}
    \frac{\partial \util[\buyer]}{\partial \allocation[\buyer][\good]}(\allocation[\buyer]^*(\state); \type[\buyer]) 
    - \lambda^*(\state) \price[\good](\state) + \mu^*_{j}(\state) \doteq 0\label{eq:foc_single_buyer_1}\\ 
     \gamma \frac{\partial}{\partial \saving[\buyer]}\mathop{\E}_{(\type^\prime, \budget^\prime, \supply^\prime) } \left[ \budgetval[\buyer]^*(\type^\prime, \budget^\prime + \saving[\buyer]^*, \supply^\prime) \right]  - \langmult[\buyer]^*(\state) + \mu_{\numgoods+1}^*(\state) \doteq 0
\end{align}
Additionally, by the KKT complimentarity conditions, we have for all $\good \in \good$, $\mu_\good^* \allocation[\buyer][\good]^* = 0$ and $\mu_{\numgoods + 1}^* \saving[\buyer]^* = 0$, which gives us: 
\begin{align}
    \frac{\partial \util[\buyer]}{\partial \allocation[\buyer][\good]}(\allocation[\buyer]^*(\state); \type[\buyer]) 
    - \lambda^*(\state) \price[\good](\state)  = 0 && \forall \good \in \goods\\
    \saving[\buyer]^*(\state) > 0 \implies \gamma \frac{\partial}{\partial \saving[\buyer]}\mathop{\E}_{(\type^\prime, \budget^\prime, \supply^\prime) } \left[ \budgetval[\buyer]^*(\type^\prime, \budget^\prime + \saving[\buyer]^*, \supply^\prime) \right] -\lambda^*(\state) = 0 && \forall \good \in \goods 
\end{align}
Re-organizing expressions, yields:
\begin{align}
    \allocation[\buyer][\good]^*(\state) > 0 \implies  \lambda^*  = \frac{\frac{\partial \util[\buyer]}{\partial \allocation[\buyer][\good]}(\allocation[\buyer]^*(\state); \type[\buyer]) 
    }{\price[\good](\state)} && \forall \good \in \goods\label{eq:implies_foc_single_buyer_1}\\
    \saving[\buyer]^*(\state) > 0 \implies \lambda^*(\state) =\gamma \frac{\partial}{\partial \saving[\buyer]}\mathop{\E}_{(\type^\prime, \budget^\prime, \supply^\prime) } \left[ \budgetval[\buyer]^*(\type^\prime, \budget^\prime + \saving[\buyer]^*, \supply^\prime) \right]   \label{eq:implies_foc_single_buyer_2}
\end{align}

Using the envelope theorem \cite{afriat1971envelope, milgrom2002envelope}, we can also compute $\frac{\partial \budgetval[\buyer]^*}{\partial \saving[\buyer]}(\state)$ as follows:
\begin{align}
    \frac{\partial \budgetval[\buyer]^*}{\partial \budget[\buyer]}(\state) = \lambda^*(\state) \
\end{align}
We note that for all states $\state \in \states$, $\frac{\partial \budgetval[\buyer]^*}{\partial \budget[\buyer]}(\state)$ is well-defined since $\lambda^*(\state)$ is uniquely defined for all states by \Cref{eq:implies_foc_single_buyer_1}.
Hence, combining the above with \Cref{eq:implies_foc_single_buyer_1} and \Cref{eq:implies_foc_single_buyer_2}, we get:
\begin{align}
    \allocation[\buyer][\good]^*(\state) > 0 \implies  \lambda^*  = \frac{\frac{\partial \util[\buyer]}{\partial \allocation[\buyer][\good]}(\allocation[\buyer]^*(\state); \type[\buyer]) 
    }{\price[\good]} && \forall \good \in \goods\\
    \saving[\buyer]^*(\state) > 0 \implies \frac{\partial \budgetval[\buyer]^*}{\partial \budget[\buyer]}(\state)  =\gamma \frac{\partial}{\partial \saving[\buyer]}\mathop{\E}_{(\type^\prime, \budget^\prime, \supply^\prime) } \left[ \budgetval[\buyer]^*(\type^\prime, \budget^\prime + \saving[\buyer]^*, \supply^\prime) \right]  \label{eq:foc_single_buyer_2}
\end{align}
Finally, going back to \Cref{eq:foc_single_buyer_1}, multiplying by $\allocation[\buyer][\good]^*(\state)$ and summing up across all $\good \in \goods$, we obtain:
\begin{align}
    &\sum_{\good \in \goods}\allocation[\buyer][\good]^*(\state)\frac{\partial \util[\buyer]}{\partial \allocation[\buyer][\good]}(\allocation[\buyer]^*(\state); \type[\buyer]) 
    - \lambda^*(\state) \price[\good](\state)\allocation[\buyer][\good]^*(\state) + \mu^*_{j} \allocation[\buyer][\good]^*(\state)= 0
\end{align}
Using Euler's theorem for homogeneous functions on the partial derivatives of the utility functions, we then have:
\begin{align}
    &\util[\buyer](\allocation[\buyer]^*(\state); \type[\buyer]) 
    - \lambda^*(\state) \sum_{\good \in \goods} \price[\good](\state)\allocation[\buyer][\good]^*(\state) + \mu^*_{j} \allocation[\buyer][\good]^*(\state)= 0
\end{align}
Additionally, the KKT Slack complementarity conditions, we have  $\lambda^*(\state) (\budget[\buyer] - \saving[\buyer]^*(\state)) = \lambda^*(\state) \sum_{\good \in \goods} \price[\good](\state)\allocation[\buyer][\good]^*(\state)$:
\begin{align}
    &\util[\buyer](\allocation[\buyer]^*(\state); \type[\buyer]) 
    - \lambda^*(\state) (\budget[\buyer] - \saving[\buyer]^*(\state)) = 0 \\
    &\lambda^*(\state) = \frac{\util[\buyer](\allocation[\buyer]^*(\state); \type[\buyer]) 
    }{\budget[\buyer] - \saving[\buyer]^*(\state)}
\end{align}
Combining the above conditions, with
\Cref{eq:implies_foc_market_1}, and adding to it
\Cref{eq:foc_single_buyer_2}, and ensuring that the KKT primal feasibility conditions hold as well, we obtain the following necessary conditions that need to hold for all states $\state \in \states$:
\begin{align}
    \allocation[\buyer][\good]^*(\state) > 0 &\implies   \frac{\frac{\partial \util[\buyer]}{\partial \allocation[\buyer][\good] }(\allocation[\buyer]^*(\state); \type[\buyer]) 
    }{\price[\good](\state)} =
    \frac{\util[\buyer](\allocation[\buyer]^*(\state); \type[\buyer]) 
    }{\budget[\buyer] - \saving[\buyer]^*(\state)} && \forall \good \in \goods\\
    \saving[\buyer]^*(\state) > 0 &\implies \frac{\partial \budgetval[\buyer]^*}{\partial \budget[\buyer]}(\state)  =\gamma \frac{\partial}{\partial \saving[\buyer]}\mathop{\E}_{(\type^\prime, \budget^\prime, \supply^\prime) } \left[ \budgetval[\buyer]^*(\type^\prime, \budget^\prime + \saving[\buyer]^*(\state), \supply^\prime) \right] 
\end{align}

If additionally the transition  probability function $\trans$ is CSD concave in $
\saving[\buyer]$, then $\budgetval[\buyer]$ is concave and the utility maximization problem is concave, which in turn implies that the above conditions are also sufficient.

\end{proof}

\thmfishermarketrecursiveeqm*

\begin{proof}[Proof of \Cref{thm:fisher_market_recursive_eqm}]
Fix a buyer $\buyer \in \buyers$. Suppose that $\statevalue[][*]$ solves the following Stochastic Stackelberg game:
\begin{align}
    \statevalue[][](\state) = \min_{\price \in \R_+^\numgoods} \max_{(\allocation, \saving) \in \R_+^{\numbuyers \times (\numgoods + 1)} : \allocation \price + \saving \leq \budget} \sum_{\good \in \goods} \supply[\good] \price[\good] + \sum_{\buyer \in \buyers} \left(\budget[\buyer] - \saving[\buyer]\right) \log(\util[\buyer](\allocation[\buyer], \type[\buyer])) \notag\\ 
    + \discount \mathop{\E}_{(\type^\prime, \budget^\prime, \supply^\prime) \sim \trans(\cdot \mid \state, 
    \saving)} \left[\statevalue[][](\type^\prime, \budget^\prime + \saving, \supply^\prime) \right]
\end{align}

Define the Lagrangian associated with the following optimization problem:
\begin{align}
    \min_{\price \in \R_+^\numgoods} \max_{(\allocation, \saving) \in \R_+^{\numbuyers \times (\numgoods + 1)} : \allocation \price + \saving \leq \budget} \sum_{\good \in \goods} \supply[\good] \price[\good] + \sum_{\buyer \in \buyers} \left(\budget[\buyer] - \saving[\buyer]\right) \log(\util[\buyer](\allocation[\buyer], \type[\buyer])) \notag\\ 
    + \discount \mathop{\E}_{(\type^\prime, \budget^\prime, \supply^\prime) \sim \trans(\cdot \mid \state, 
    \saving)} \left[\statevalue[][*](\type^\prime, \budget^\prime + \saving, \supply^\prime) \right]
\end{align}
as follows:
\begin{align}
    \lang(\price, \allocation, \saving, \langmult) &= \sum_{\good \in \goods} \supply[\good]  \price[\good] + \sum_{\buyer \in \buyers} \left(\budget[\buyer] - \saving[\buyer]\right) \log(\util[\buyer](\allocation[\buyer], \type[\buyer])) + \discount \mathop{\E}_{(\state^\prime) \sim \trans(\cdot \mid \state, 
    \saving)} \left[\statevalue[][](\type^\prime, \budget^\prime + \saving, \supply^\prime) \right] \notag \\
    &+ \sum_{\good \in \goods} \langmult[\buyer] \left( \budget[\buyer] - \allocation[\buyer] \cdot \price + \saving[\buyer]\right) \enspace .
\end{align}
By \Cref{thm:necessary_conditions}, the necessary optimality conditions for the stochastic Stackelberg game 
are that for all states $\state \in \states$ there exists $\bm \mu^*(\state) \in \R^{\numbuyers \times (\numgoods + 1)}$, and $\bm \nu^* (\state) \in \R^{\numbuyers \times (\numgoods + 1)}_+$ associated with the non-negativity constraints for $(\allocation(\state), \saving(\state))$, and $\price$ respectively, and $\langmult^*(\state) \in \R^\numgoods_+$ associated with the spending constraint for $(\allocation(\state), \saving(\state))$ such that:
\begin{align}
    &\supply[\good]  - \sum_{\buyer \in \buyers} \langmult[\buyer]^*(\state) \allocation[\buyer][\good]^*(\state) - \nu_\good^*(\state) \doteq 0 \label{eq:foc_rfm_1} && \forall \good \in \goods\\
    &\frac{\budget[\buyer] - \saving[\buyer]^*(\state)}{\util[\buyer](\allocation[\buyer]^*(\state))} \frac{\partial\util[\buyer]}{\partial \allocation[\buyer][\good]} (\allocation[\buyer]^*(\state); \type[\buyer])  
    - \langmult[\buyer]^*(\state) \price[\good](\state) + \mu_{\buyer \good}^*(\state) \doteq 0 && \forall \buyer \in \buyers, \good \in \goods \label{eq:foc_rfm_2}\\
    &- \log\left(\util[\buyer](\allocation[\buyer]^*(\state)) \right) +\discount  \frac{\partial }{\partial \saving[\buyer]}\mathop{\E}_{(\type^\prime, \budget^\prime, \supply^\prime)} \left[\statevalue[][](\type^\prime, \budget^\prime + \saving^*(\state), \supply^\prime) \right] - \langmult[\buyer]^*(\state) + \mu_{\buyer (\numgoods + 1)}^*(\state) \doteq 0\label{eq:foc_rfm_3} && \forall \buyer \in \buyers
\end{align}

Note that by \Cref{thm:necessary_conditions}, we also have $\mu_{\buyer (\numgoods + 1)}^*(\state) \saving[\buyer]^*(\state) = \mu_{\buyer  + 1)}^*(\state) \allocation[\buyer][\good]^*(\state) = 0$ which gives us:

\begin{align}
    \price[\good](\state) > 0 
    &\implies \supply[\good]  - \sum_{\buyer \in \buyers} \langmult[\buyer]^*(\state) \allocation[\buyer][\good]^*(\state) =  0  && \forall \good \in \goods\\
    \allocation[\buyer][\good]^*(\state) > 0 
    &\implies \frac{\budget[\buyer] - \saving[\buyer]^*(\state)}{\util[\buyer](\allocation[\buyer]^*(\state))} \frac{\partial\util[\buyer]}{\partial \allocation[\buyer][\good]} (\allocation[\buyer]^*(\state)) 
    - \langmult[\buyer]^*(\state) \price[\good](\state)  = 0 && \forall \buyer \in \buyers, \good \in \goods\\
    \saving[\buyer]^*(\state) > 0 
    &\implies -\log\left(\util[\buyer](\allocation[\buyer]^*(\state); \type[\buyer]) \right) +\discount  \frac{\partial }{\partial \saving[\buyer]}\mathop{\E}_{(\type^\prime, \budget^\prime, \supply^\prime)} \left[\statevalue[][](\type^\prime, \budget^\prime + \saving^*(\state), \supply^\prime) \right] \nonumber\\
    &- \langmult[\buyer]^*(\state) + \mu_{\buyer (\numgoods + 1)}^*(\state) = 0 && \forall \buyer \in \buyers
\end{align}

Re-organizing expressions, we obtain:
\begin{align}
    \price[\good](\state) > 0 
    &\implies \supply[\good]   =  \sum_{\buyer \in \buyers} \langmult[\buyer]^*(\state) \allocation[\buyer][\good]^*(\state) && \forall \good \in \goods \label{eq:implies_foc_market_1}\\
    \allocation[\buyer][\good]^*(\state) > 0 
    &\implies 
    \frac{\util[\buyer](\allocation[\buyer]^*(\state))}{\budget[\buyer] - \saving[\buyer]^*(\state)}\langmult[\buyer]^*(\state)  = \frac{ \frac{\partial\util[\buyer]}{\partial \allocation[\buyer][\good]} (\allocation[\buyer]^*(\state)) 
    }{\price[\good](\state) } && \forall \buyer \in \buyers, \good \in \goods \label{eq:implies_foc_market_2}\\
    \saving[\buyer]^*(\state) > 0 
    &\implies - \log\left(\util[\buyer](\allocation[\buyer]^*(\state)) \right) +\discount  \frac{\partial }{\partial \saving[\buyer]}\mathop{\E}_{(\type^\prime, \budget^\prime, \supply^\prime)} \left[\statevalue[][](\type^\prime, \budget^\prime + \saving^*(\state), \supply^\prime) \right] \nonumber\\
    &-\langmult[\buyer]^*(\state)  = 0 && \forall \buyer \in \buyers\label{eq:implies_foc_market_3}
\end{align}

Using the envelope theorem, we can compute $\frac{\partial \statevalue[][]}{\partial \budget[\buyer]}$ as follows:
\begin{align}
    \frac{\partial \statevalue[][*]}{\partial \budget[\buyer]}(\state) = \log\left(\util[\buyer](\allocation[\buyer]^*(\state); \type[\buyer]) \right) + \langmult[\buyer]^*(\state)
\end{align}
Once again note that $\frac{\partial \statevalue[][]}{\partial \budget[\buyer]}$ is well-defined by \Cref{eq:implies_foc_market_2}.

Re-organizing expressions, we get:
\begin{align}
    \langmult[\buyer]^*(\state) = \frac{\partial \statevalue[][]}{\partial \budget[\buyer]}(\state) - \log\left(\util[\buyer](\allocation[\buyer]^*(\state); \type[\buyer]) \right)\label{eq:envelope_value_rfm}
\end{align}

Combining \Cref{eq:envelope_value_rfm} and \Cref{eq:implies_foc_market_3}, we obtain:
\begin{align}
    \saving[\buyer]^*(\state) > 0 
    &\implies - \log\left(\util[\buyer](\allocation[\buyer]^*(\state); \type[\buyer]) \right) +\discount  \frac{\partial }{\partial \saving[\buyer]}\mathop{\E}_{(\type^\prime, \budget^\prime, \supply^\prime)} \left[\statevalue[][](\type^\prime, \budget^\prime + \saving^*(\state), \supply^\prime) \right] \nonumber\\
     &- \frac{\partial \statevalue[][]}{\partial \budget[\buyer]}(\state) + \log\left(\util[\buyer](\allocation[\buyer]^*(\budget); \type[\buyer]) \right)  = 0 && \forall \buyer \in \buyers\\
    \saving[\buyer]^*(\state) > 0 
    &\implies \discount  \frac{\partial }{\partial \saving[\buyer]}\mathop{\E}_{(\type^\prime, \budget^\prime, \supply^\prime)} \left[\statevalue[][](\type^\prime, \budget^\prime + \saving^*(\state), \supply^\prime) \right] - \frac{\partial \statevalue[][]}{\partial \budget[\buyer]}(\state)  = 0 && \forall \buyer \in \buyers\\
    \saving[\buyer]^*(\state) > 0 
    &\implies  \frac{\partial \statevalue[][]}{\partial \budget[\buyer]}(\state) = \discount  \frac{\partial }{\partial \saving[\buyer]}\mathop{\E}_{(\type^\prime, \budget^\prime, \supply^\prime)} \left[\statevalue[][](\type^\prime, \budget^\prime + \saving^*(\state), \supply^\prime) \right] && \forall \buyer \in \buyers\label{eq:implies_intertime_foc_market}
\end{align}

Going back to \Cref{eq:foc_rfm_2}, multiplying both sides by $\allocation[\buyer][\good]^*(\state)$ and summing up across all $\good \in \goods$, we get:
\begin{align}
    \sum_{\good \in \goods}\frac{\budget[\buyer] - \saving[\buyer]^*(\state)}{\util[\buyer](\allocation[\buyer]^*(\state); \type[\buyer])} \sum_{\good \in \goods} \allocation[\buyer][\good]^*(\state)\frac{\partial\util[\buyer]}{\partial \allocation[\buyer][\good]} (\allocation[\buyer]^*(\state))  - \langmult[\buyer]^*(\state) \sum_{\good \in \goods}\price[\good](\state) \allocation[\buyer][\good]^*(\state) = 0 \\
    \frac{\budget[\buyer] - \saving[\buyer]^*(\state)}{\util[\buyer](\allocation[\buyer]^*(\state); \type[\buyer])} \util[\buyer](\allocation[\buyer]^*(\state); \type[\buyer])  - \langmult[\buyer]^*(\state) \sum_{\good \in \goods}\price[\good](\state) \allocation[\buyer][\good]^*(\state) = 0 && \text{(Euler's Theorem)}\\
    \budget[\buyer] - \saving[\buyer]^*(\state)  - \langmult[\buyer]^*(\state) \sum_{\good \in \goods}\price[\good](\state) \allocation[\buyer][\good]^*(\state) = 0
\end{align}

By \Cref{thm:necessary_conditions}, we have that $\langmult[\buyer]^*(\state) \left( \budget[\buyer] - \sum_{\good \in \goods}\price[\good](\state) \allocation[\buyer][\good]^*(\state) - \saving[\buyer]^*(\state)  \right) = 0$, which gives us:
\begin{align}
    \budget[\buyer] - \saving[\buyer]^*(\state)  - \langmult[\buyer]^*(\state) (\budget[\buyer] - \saving[\buyer]^*(\state) )= 0\\
    \langmult[\buyer]^*(\state)  = 1
\end{align}

Combining the above with \Crefrange{eq:implies_foc_market_1}{eq:implies_foc_market_3} we obtain:

\begin{align}
    &\price[\good](\state) > 0 \implies \supply[\good]   =  \sum_{\buyer \in \buyers} \allocation[\buyer][\good]^*(\state) && \forall \good \in \goods \\
    &\allocation[\buyer][\good]^*(\state) > 0 \implies \frac{\util[\buyer](\allocation[\buyer]^*(\state))}{\budget[\buyer] - \saving[\buyer]^*(\state)} = \frac{\frac{\partial\util[\buyer]}{\partial \allocation[\buyer][\good]} (\allocation[\buyer]^*(\state); \type[\buyer])}{\price[\good](\state) } && \forall \buyer \in \buyers, \good \in \goods \\
    &\saving[\buyer]^*(\state) > 0 \implies  \frac{\partial \statevalue[][]}{\partial \budget[\buyer]}(\budget) = \discount  \frac{\partial }{\partial \saving[\buyer]}\mathop{\E}_{(\type^\prime, \budget^\prime, \supply^\prime)} \left[\statevalue[][](\type^\prime, \budget^\prime + \saving^*(\state), \supply^\prime) \right] && \forall \buyer \in \buyers
\end{align}
Since the utility functions are non-satiated, and by the second equation, the buyers are utility maximizing at state $\state$ over all allocations, we must also have that Walras' law holds, i.e., $\price \cdot \left( \supply - \sum_{\buyer \in \buyers} \allocation[\buyer]\right) - \sum_{\buyer \in \buyers} \saving[\buyer] = 0$. Walras' law combined with the first equation above then imply the second condition of a recursive competitive equilibrium. Finally, by \Cref{lemma:savings_foc}, the last two equations imply the first condition of recursive competitive equilibrium.


\end{proof}

%% file: appendix/experiments.tex
\section{Experiment Details}\label{sec_app:experiments}
\subsection{Stochastic Fisher market without interest rates}
For the without interest rates setup, we initialized a stochastic Fisher market with $\numbuyers = 2$ buyers and $\numgoods = 2$ goods. To simplify the analysis, we assumed deterministic transitions such that the buyers get a constant new budgets of $9.5$ at each time period, and their types/valuations as well as the supply of goods does not change at each state, i.e., the type/valuation space and supply space has cardinality 1. This reduced the market to a deterministic repeated market setting in which the amount of budget saved by the buyers differentiates different states of the market. To initialize the state space of the market, we first fixed a range of $[10,50]^\numgoods$ for the buyers' valuations and drew for all buyers $\buyer \in \buyers$ valuations $\valuation[\buyer]$ from that range uniformly at random at the beginning of the experiment. (We scaled the valuations differently for different markets to ensure positive utilities though.) We have assumed the supply of goods is $\ones[\numgoods]$ and that the budget space was $[9,10]^\numbuyers$. This means that our state space for our experiments was $\states = \{ (\valuation[1], \valuation[2])\} \times \{ \ones[\numgoods] \} \times [9, 10]^\numbuyers$. We note that although the assumption that buyers valuations/type space has cardinality one does simplify the problem, the supply of the goods being $\ones$ at each state is wlog because goods are divisible and the allocation of goods to buyers at each state can then be interpreted as the percentage of a particular good allocated to a buyer. We assumed initial budgets of $\budget[][0] = 10_{\numbuyers}$ for buyers. 

Since the state space is continuous, the value function has continuous domain in the stochastic Fisher market setting. As a result, we had to use fitted variant of value iteration. In particular, we assumed that the value function had a linear form at each state such that $\statevalue(\type, \budget, \supply; \a, c) = \a^T\budget + c$ for some parameters $\a \in \R^\numbuyers, c \in \R$, and we tried to approximate the value function at the next step of value iteration by using linear regression. That is, at each value iteration step, we uniformly sampled 25 budget vectors from the range $[9,10]^\numbuyers$. Next, for each sampled budget $\budget$, we solved the min-max step given that budget as a state. This process gave us (budget, value) pairs on which we ran linear regression to approximate the value function at the next iterate. 

To solve the generalized min-max operator at each step of value iteration, we used the \mydef{nested gradient descent ascent (GDA)} \cite{goktas2021minmax} (\Cref{alg:nested_gda_on_qfunc}) along with JAX gradients, which is not guaranteed to converge to a global optimum since the min-max Stackelberg game for stochastic Fisher markets is convex-non-concave. Then, we have run value iteration for $30$ iterations. We ran nested GDA with learning rates $\learnrate[\allocation] = 1.4$, $\learnrate[\price] = 1.5\times10^{-2}$ for linear, $\learnrate[\allocation] = 1.5$, $\learnrate[\price] = 6.5\times 10^{-4}$ for leontief, and $\learnrate[\allocation] = 1.4$, $\learnrate[\price] = 5\times 10^{-3}$ for Cobb-Douglas. The outer loop of nested GDA was run for $\numiters_{\price} = 60$ iterations, while its inner loop was run for $\numiters_{\allocation} = 100$ iterations, and break from the nested GDA if we obtain an excess demand with norm lower than $0.01$. We depict the trajectory of the average value of the value function at each iteration of value iteration under nested GDA in \Cref{fig:avg_values_small}.

\subsection{Stochastic Fisher market with interest rates}
For the with interest rates setup, we initialized a stochastic Fisher market with $\numbuyers = 5$ buyers and $\numgoods = 5$ goods. This time, we implemented a stochastic transitions. Though buyers still get a constant new budgets of $9.5$ at each time step, and their types/valuations as well as the supply of goods does not change at each state, their savings from last time step may decrease or increase according to some probabilistic interest rates. In specific, at each time step, we consider five interest rates $\{0.9, 1.0, 1.1, 1.2, 1.5\}$, each with probability $0.2$.
Thus, we have a stochastic market setting in which the amount of budget possessed by the buyers at the beginning of each time step differentiates different states of the market. To initialize the state space of the market, we first fixed a range of $[0,1]^\numgoods$ for the buyers' valuations and drew for all buyers $\buyer \in \buyers$ valuations $\valuation[\buyer]$ from that range uniformly at random at the beginning of the experiment. (We scaled the valuations differently for different markets to ensure positive utilities though.)
We have assumed the supply of goods is $\ones[\numgoods]$ and that the budget space was $[9,10]^\numbuyers$. This means that our state space for our experiments was $\states = \{ (\valuation[1], \valuation[2], \valuation[3], \valuation[4], \valuation[5])\} \times \{ \ones[\numgoods] \} \times [9, 10]^\numbuyers$. We note that although the assumption that buyers valuations/type space has cardinality one does simplify the problem, the supply of the goods being $\ones$ at each state is wlog because goods are divisible and the allocation of goods to buyers at each state can then be interpreted as the percentage of a particular good allocated to a buyer. We assumed initial budgets of $\budget[][0] = 10_{\numbuyers}$ for buyers. 

Since the state space is continuous, the value function has continuous domain in the stochastic Fisher market setting. As a result, we had to use fitted variant of value iteration. In particular, we assumed that the value function had a linear form at each state such that $\statevalue(\type, \budget, \supply; \a, c) = \a^T\budget + c$ for some parameters $\a \in \R^\numbuyers, c \in \R$, and we tried to approximate the value function at the next step of value iteration by using linear regression. That is, at each value iteration step, we uniformly sampled 25 budget vectors from the range $[9,10]^\numbuyers$. Next, for each sampled budget $\budget$, we solved the min-max step given that budget as a state. This process gave us (budget, value) pairs on which we ran linear regression to approximate the value function at the next iterate. 

To solve the generalized min-max operator at each step of value iteration, we used the \mydef{nested gradient descent ascent (GDA)} \cite{goktas2021minmax} (\Cref{alg:nested_gda_on_qfunc}) along with JAX gradients, which is not guaranteed to converge to a global optimum since the min-max Stackelberg game for stochastic Fisher markets is convex-non-concave. Then, we have run value iteration for $30$ iterations. We ran nested GDA with learning rates $\learnrate[\allocation] = 1.7$, $\learnrate[\price] = 2\times10^{-2}$ for linear, $\learnrate[\allocation] = 2$, $\learnrate[\price] = 5\times 10^{-5}$ for leontief, and $\learnrate[\allocation] = 1.8$, $\learnrate[\price] = 2.5\times 10^{-2}$ for Cobb-Douglas. The outer loop of nested GDA was run for $\numiters_{\price} = 60$ iterations, while its inner loop was run for $\numiters_{\allocation} = 100$ iterations, and break from the nested GDA if we obtain an excess demand with norm lower than $0.01$. We depict the trajectory of the average value of the value function at each iteration of value iteration under nested GDA in \Cref{fig:avg_values_big}. 

\subsection{Other Details}
\paragraph{Programming Languages, Packages, and Licensing}
We ran our experiments in Python 3.7 \cite{van1995python}, using NumPy \cite{numpy}, CVXPY \cite{diamond2016cvxpy}, and JAX \cite{jax2018github}.
\Cref{fig:avg_values_small} and \Cref{fig:avg_values_big} were graphed using Matplotlib \cite{matplotlib}.

Python software and documentation are licensed under the PSF License Agreement. Numpy is distributed under a liberal BSD license.  Matplotlib only uses BSD compatible code, and its license is based on the PSF license. CVXPY is licensed under an APACHE license. 

\paragraph{Implementation Details}

In our execution of \Cref{alg:nested_gda_on_qfunc}, in order to project each allocation computed onto the budget set of the consumers, i.e., $\{\allocation \in \R^{\numbuyers \times \numgoods}_+ \mid \allocation\price \leq \budget\}$, we used the CVXPY with MOSEK solver with warm start option, a feature that enables the solver to exploit work from previous solves. 
\paragraph{Computational Resources}
Our experiments were run on Google Colab with 12.68GB RAM, and took about 8 hours to run the Stochastic Fisher market without interest rates experiment (for each utility function class) and about 8.5 hours to run the Stochastic Fisher market with interest rates experiment (for each utility function class). Only CPU resources were used.

\paragraph{Code Repository}
The data our experiments generated, and the code used to produce our visualizations, can be found in our code repository ({\color{blue}\rawcoderepo}).